\documentclass[letterpaper,11pt]{article}
\usepackage[margin=1in]{geometry}
\usepackage[T1]{fontenc}
\usepackage{amsmath,amssymb,amsthm,mathtools}
\usepackage{newtxtext,newtxmath}
\usepackage[ruled,vlined]{algorithm2e}
\SetAlgoSkip{}
\SetAlCapSkip{0.2em}
\SetInd{0.35em}{0.6em}
\usepackage{booktabs}
\usepackage{placeins}
\usepackage{graphicx}
\usepackage[numbers,sort&compress]{natbib}
\usepackage{enumitem}
\setlist{topsep=1.5pt plus 1pt minus 0.5pt,itemsep=0.5pt plus 0.5pt,parsep=0pt,leftmargin=*}
\usepackage{titlesec}
\titleformat{\section}{\normalfont\large\bfseries}{\thesection}{0.6em}{}
\titleformat{\subsection}{\normalfont\normalsize\bfseries}{\thesubsection}{0.6em}{}
\titleformat{\subsubsection}{\normalfont\normalsize\itshape}{\thesubsubsection}{0.6em}{}
\titlespacing*{\section}{0pt}{1.35ex plus .45ex minus .2ex}{0.65ex plus .2ex}
\titlespacing*{\subsection}{0pt}{1.0ex plus .35ex minus .2ex}{0.45ex plus .15ex}
\titlespacing*{\subsubsection}{0pt}{0.8ex plus .25ex minus .15ex}{0.35ex plus .1ex}
\usepackage[colorlinks=true,allcolors=black]{hyperref}
\usepackage[capitalise,nameinlink]{cleveref}
\providecommand{\doi}[1]{\href{https://doi.org/#1}{doi: \nolinkurl{#1}}}

\crefname{algocf}{algorithm}{algorithms}
\Crefname{algocf}{Algorithm}{Algorithms}
\SetKwInput{KwIn}{Input}
\SetKwInput{KwOut}{Output}
\SetKwFor{For}{for}{do}{}
\SetKwIF{If}{ElseIf}{Else}{if}{then}{else if}{else}{}

\AtBeginDocument{%
  \setlength{\abovedisplayskip}{6pt plus 2pt minus 2pt}%
  \setlength{\belowdisplayskip}{6pt plus 2pt minus 2pt}%
  \setlength{\abovedisplayshortskip}{4pt plus 2pt minus 1pt}%
  \setlength{\belowdisplayshortskip}{4pt plus 2pt minus 1pt}%
}

\makeatletter
\def\thm@space@setup{%
  \thm@preskip=5pt plus 1pt minus 1pt%
  \thm@postskip=5pt plus 1pt minus 1pt%
}
\renewcommand{\maketitle}{%
  \begingroup
  \begin{center}
  {\LARGE\bfseries \@title\par}
  \vspace{0.55em}
  {\normalsize \@author\par}
  \end{center}
  \vspace{-0.55em}
  \endgroup
}
\makeatother
\renewenvironment{abstract}{%
  \begin{center}{\bfseries Abstract}\end{center}%
  \vspace{-0.6em}\noindent
}{\par\vspace{0.35em}}

\newtheorem{theorem}{Theorem}[section]
\newtheorem{lemma}{Lemma}[section]

\newtheorem{corollary}{Corollary}[section]
\newtheorem{claim}{Claim}[section]
\theoremstyle{definition}

\theoremstyle{remark}
\newtheorem{remark}{Remark}[section]
\theoremstyle{plain}

\title{Reducing Prize-Collecting Stroll and Related Routing Problems to Prize-Collecting TSP}

\author{%
Hong Li\\
\small School of Mathematics and Statistics, Yunnan University\\
\small \texttt{honglimath@126.com}
}
\date{}
\begin{document}

\maketitle

\begin{abstract}\small
We consider prize-collecting variants of the metric traveling salesperson problem (TSP). The prize-collecting stroll problem (PCS) is the path version of the prize-collecting TSP (PCTSP). Given a complete metric graph, two distinct prescribed terminal vertices $s,t$, and nonnegative penalties on vertices, PCS asks for an $s$-$t$ tour minimizing its length plus the total penalty of vertices that are not visited by it. We study a common generalization of PCTSP, PCS, and several related prize-collecting routing problems, which we call the prize-collecting $\Phi$-TSP. In this model, an interface $\Phi$ specifies a set of prescribed vertices together with their parity and connectivity requirements. We show that, if a polynomial-time $\rho$-approximation algorithm for PCTSP is available, then, for every fixed $\varepsilon>0$, there is a polynomial-time $(\rho+\varepsilon)$-approximation algorithm for the prize-collecting $\Phi$-TSP when the number of prescribed vertices is bounded by a fixed constant. Consequently, PCS can be approximated as well as PCTSP up to an arbitrarily small additive loss, closing the gap between the best-known approximation guarantees for the two problems.
\end{abstract}

\section{Introduction}
\label{sec:intro}

\subsection{PC-\texorpdfstring{$\Phi$}{Phi}-TSP and related prize-collecting routing problems}

In a complete metric graph $G=(V,E)$ with edge lengths $\ell:E\to\mathbb R_{\ge 0}$, the metric traveling salesperson problem (TSP) asks for a shortest tour visiting all vertices. Given nonnegative vertex penalties $\pi:V\to\mathbb R_{\ge 0}$ and a root $r\in V$, the prize-collecting TSP (PCTSP) asks for a possibly empty tour rooted at $r$ that minimizes its length plus the total penalty of vertices not visited by it. Given two distinct terminals $s,t\in V$, the prize-collecting stroll problem (PCS) asks for an $s$-$t$ tour that minimizes its length plus the total penalty of vertices not visited by it. Inspired by the work of Traub, Vygen, and Zenklusen on $\Phi$-TSP~\cite{TVZ2022}, we describe tours and $s$-$t$ tours as edge multisets rather than as walks and study the following common generalization.

A prize-collecting $\Phi$-TSP (PC-$\Phi$-TSP) instance consists of $(G,\ell,\pi)$ as above and an interface $\Phi=(I,Q,\mathcal P)$. Here $I\subseteq V$ is a set of prescribed vertices, $Q\subseteq I$ is a set of even cardinality whose vertices are prescribed to have odd degree, and $\mathcal P$ is a partition of $I$ that specifies the required connectivity among the prescribed vertices. The problem asks for a multiset $F$ of edges of $G$. Let $V(F)$ be the set of vertices incident to $F$, and regard the vertices in $V(F)\cup I$ as visited by $F$. The multigraph $(V(F)\cup I,F)$ is required to satisfy: \textbf{(i) parity constraint:} the vertices in $Q$ have odd degree and all vertices in $(V(F)\cup I)\setminus Q$ have even degree; \textbf{(ii) $I$-connectivity constraint:} every connected component contains at least one vertex of $I$; and \textbf{(iii) $\mathcal P$-connectivity constraint:} for every $C\in\mathcal P$, all vertices of $C$ lie in the same connected component. The objective is to minimize $\sum_{e\in F}\ell(e)+\sum_{v\in V\setminus (V(F)\cup I)}\pi(v)$. Figure~\ref{fig:pc_phi_example} illustrates PC-$\Phi$-TSP.

\begin{figure}[!htbp]\small 
    \centering
    \includegraphics[height=0.25\textheight,keepaspectratio]{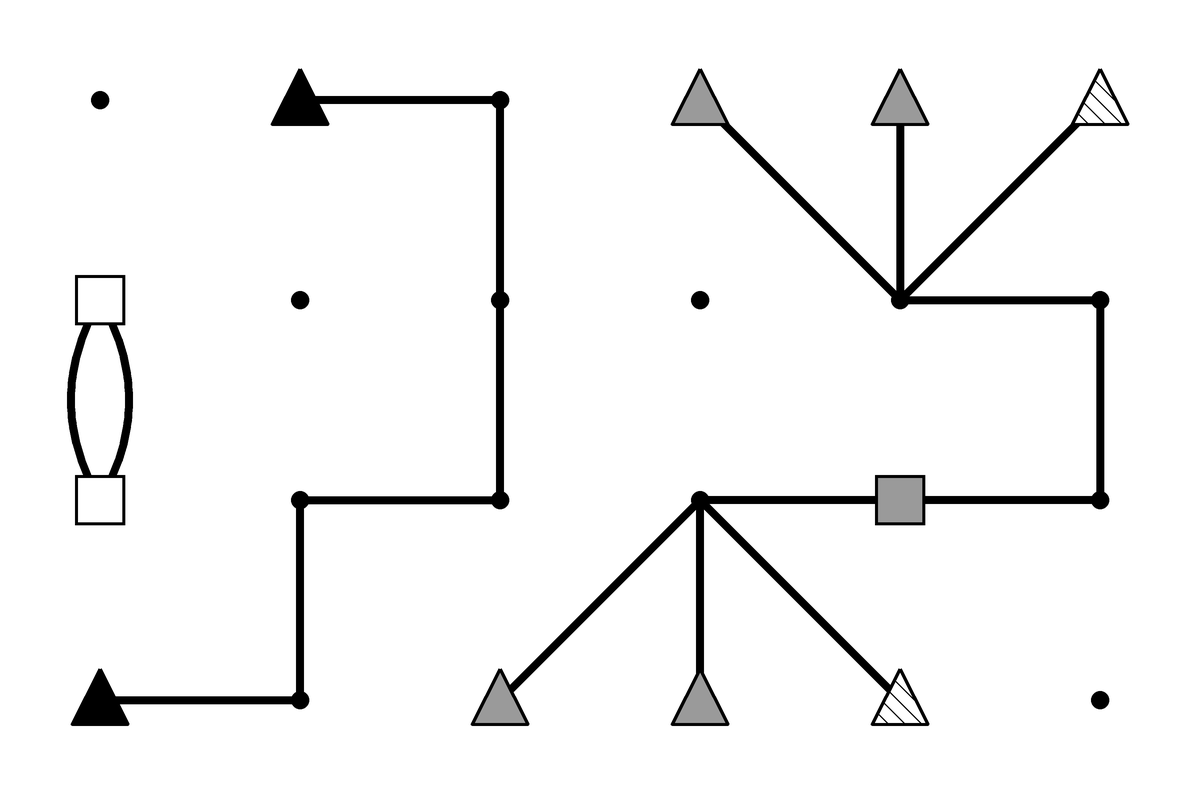}
    \caption{A feasible solution to a PC-$\Phi$-TSP instance. Squares are vertices in $I\setminus Q$, triangles are vertices in $Q$, and vertices in $I$ with the same fill style are contained in the same set $C\in\mathcal P$.}
    \label{fig:pc_phi_example}
\end{figure}

PCTSP is exactly PC-$\Phi$-TSP with $I=\{r\}$, $Q=\emptyset$, and $\mathcal P=\{\{r\}\}$, and PCS is exactly PC-$\Phi$-TSP with $I=Q=\{s,t\}$ and $\mathcal P=\{\{s,t\}\}$. PC-$\Phi$-TSP also generalizes the prize-collecting connected $T$-join problem~\cite{CheriyanFriggstadGao2015} and a prize-collecting analogue of the multiple TSP~\cite{XuRodrigues2015}.

\subsection{Related work}

PCTSP is a special case of a more general problem introduced by Balas~\cite{Balas1989}. Bienstock, Goemans, Simchi-Levi, and Williamson~\cite{BGW1993} gave a $2.5$-approximation based on threshold rounding, and Goemans and Williamson~\cite{GW1995} obtained a $2$-approximation by a primal-dual method. After subsequent improvements~\cite{ArcherBHK2011,Goemans2009,BlauthNagele2023}, Blauth, Klein, and N{\"a}gele~\cite{BKN2026} gave the previously best-known approximation ratio, $1.599$. More recently, we improved this guarantee to $1.55576$~\cite{Li2026PCTSP} through a stronger analysis of a simpler version of their algorithm.

PCS is the path version of PCTSP and is also known as the prize-collecting path problem. A modification of the primal-dual approach of Goemans and Williamson~\cite{GW1995} gives a $2$-approximation~\cite{ChaudhuriGRT2003}. This approximation ratio was improved~\cite{ArcherBHK2011,AnKS2015}, and Blauth, Klein, and N{\"a}gele~\cite{BKN2026} obtained the previously best-known approximation ratio, $1.6662$. Archer and Blasiak~\cite{ArcherBlasiak2010} also used a variant of PCS in which only one terminal is fixed, as a tool in approximation algorithms for the minimum latency problem. Alimi, M{\"o}mke, and Ruderer~\cite{AlimiMoemkeRuderer2025} considered a multi-path version of PCS, called the prize-collecting multi-path TSP, and gave a $2.41$-approximation algorithm.

Our work is also connected to the relation between TSP and the path TSP, which asks for a shortest $s$-$t$ tour visiting all vertices for distinct terminals $s,t$. Christofides~\cite{Christofides1976} gave the classical $1.5$-approximation for TSP; Karlin, Klein, and Oveis Gharan~\cite{KarlinKleinOveisGharan2024,KarlinKleinOveisGharanIPCO2023} improved the approximation ratio slightly below $1.5$, and Gurvits, Klein, and Leake~\cite{GurvitsKleinLeake2024} gave a stronger analysis of their algorithm. For the path TSP, Hoogeveen~\cite{Hoogeveen1991} showed that the natural adaptation of Christofides' algorithm has approximation ratio $5/3$. After improvements~\cite{AnKS2015,Sebo2013,Vygen2016,GottschalkVygen2018,SeboVanZuylen2019,TraubVygen2019}, Zenklusen~\cite{Zenklusen2019} gave a $1.5$-approximation.

Traub, Vygen, and Zenklusen~\cite{TVZ2022} introduced $\Phi$-TSP, a generalization of TSP and the path TSP in which all vertices must be visited and the interface $\Phi=(I,Q,\mathcal P)$ has the same form as in the definition of PC-$\Phi$-TSP. They showed that, when $|I|$ is bounded by a fixed constant, any $\alpha$-approximation algorithm for TSP yields, for every fixed $\varepsilon>0$, an $(\alpha+\varepsilon)$-approximation algorithm for $\Phi$-TSP. Their result implies that the path TSP can be approximated as well as TSP, up to an arbitrarily small additive loss. Combined with the algorithm of Karlin, Klein, and Oveis Gharan~\cite{KarlinKleinOveisGharan2024,KarlinKleinOveisGharanIPCO2023}, this gives a better-than-$1.5$-approximation for the path TSP. Traub and Vygen~\cite[Sections~16.3--16.7]{TraubVygen2024} extended the reduction from the path TSP to TSP to a reduction from a multi-path version of the path TSP with a fixed number of paths to TSP.

However, the analogous relation between PCTSP and PCS was unknown. Even after the work of Traub, Vygen, and Zenklusen~\cite{TVZ2022}, the approximation ratios of Blauth, Klein, and N{\"a}gele~\cite{BKN2026} for these two problems remained separated: $1.599$ for PCTSP and $1.6662$ for PCS. Our result closes this gap: PCS can be approximated as well as PCTSP up to an arbitrarily small additive loss.

\subsection{Our results}

Our main result shows that, for interfaces $\Phi=(I,Q,\mathcal P)$ with $|I|$ bounded by a fixed constant, PC-$\Phi$-TSP can be approximated as well as PCTSP up to an arbitrarily small additive loss in the approximation ratio.

\begin{theorem}
\label{thm:main}
Suppose that there is a polynomial-time $\rho$-approximation algorithm for PCTSP. Then, for every fixed positive integer $k$ and every fixed $\varepsilon>0$, there is a polynomial-time $(\rho+\varepsilon)$-approximation algorithm for PC-$\Phi$-TSP on instances with interface $\Phi=(I,Q,\mathcal P)$ satisfying $|I|\le k$.
\end{theorem}

The polynomial-time guarantee in Theorem~\ref{thm:main} is for fixed $k$ and fixed $\varepsilon$: the running time is of the form $n^{f(k,
1/\varepsilon)}$ for a function $f$ depending on $\rho$, times the running time of the PCTSP algorithm, where $n=|V|$; see the proof of Theorem~\ref{thm:main}. Since PCS is the special case with $I=\{s,t\}$, and hence $|I|=2$, applying Theorem~\ref{thm:main} to PCS gives the following consequence.

\begin{corollary}
\label{cor:pcs}
Suppose that PCTSP admits a polynomial-time $\rho$-approximation algorithm. Then, for every fixed $\varepsilon>0$, PCS admits a polynomial-time $(\rho+\varepsilon)$-approximation algorithm.
\end{corollary}

Using the $1.599$-approximation algorithm for PCTSP of Blauth, Klein, and N{\"a}gele~\cite{BKN2026}, Theorem~\ref{thm:main} already gives a better-than-$1.6$-approximation algorithm for PC-$\Phi$-TSP whenever $|I|$ is bounded by a fixed constant. In particular, it improves the best-known guarantee for PCS from $1.6662$ to below $1.6$. Further improvements for PCTSP, such as the one in~\cite{Li2026PCTSP}, transfer in the same way. Theorem~\ref{thm:main} also gives the same guarantees for the prize-collecting connected $T$-join problem with fixed $|T|$~\cite{CheriyanFriggstadGao2015} and for a prize-collecting analogue of the multiple TSP with a fixed number of depots~\cite{XuRodrigues2015}.

Given terminal pairs $\{s_1,t_1\},\ldots,\{s_\kappa,t_\kappa\}$, the prize-collecting multi-path TSP~\cite{AlimiMoemkeRuderer2025} asks for one $s_i$-$t_i$ tour for each $i$, minimizing their total length plus the total penalty of vertices visited by none of them. Although this problem is not a special case of PC-$\Phi$-TSP (the union of the tours does not determine how the terminals are paired), Theorem~\ref{thm:main} also yields the same transfer for it when the number of terminal pairs is fixed.

\begin{corollary}
\label{cor:multipath}
Suppose that PCTSP admits a polynomial-time $\rho$-approximation algorithm. Then, for every fixed positive integer $\kappa$ and every fixed $\varepsilon>0$, the prize-collecting multi-path TSP with at most $\kappa$ terminal pairs admits a polynomial-time $(\rho+\varepsilon)$-approximation algorithm.
\end{corollary}

Appendix~\ref{app:multipath} proves the corollary by adapting the multi-path TSP framework of Traub and Vygen~\cite[Sections~16.4--16.7]{TraubVygen2024}: we use Theorem~\ref{thm:main} where they apply a TSP algorithm directly, and we add penalties to their dynamic program.

\subsection{Overview of the approach}

\textbf{The approach of Traub, Vygen, and Zenklusen.} Our approach is inspired by the work of Traub, Vygen, and Zenklusen on $\Phi$-TSP~\cite{TVZ2022}. They start from a constant-factor approximation algorithm for $\Phi$-TSP and improve its approximation ratio iteratively. In each iteration, given an approximation algorithm for TSP and the current approximation algorithm for $\Phi$-TSP, they compute two candidates and return the better one. The first candidate starts from a union of TSP tours that satisfies their analogues of the $I$-connectivity and $\mathcal P$-connectivity constraints but not necessarily the parity constraint, and repairs parity by adding a minimum $T$-join with $T=Q$, the set of vertices prescribed to have odd degree. This candidate is good when the parity repair cost, that is, the length of this $T$-join, is small. The second candidate is computed by a dynamic program over the laminar family from the minimum $T$-join problem. The dynamic program guesses some sets in this family that an optimum solution crosses only a few times, together with the crossing edges, and constructs local $\Phi$-TSP instances, whose prescribed vertices include the endpoints of the guessed edges, whose parity requirements follow from these edges, and whose connectivity requirements are also guessed; it solves the local instances using the current algorithm and combines the resulting solutions with the guessed edges. Since the guessed edges are paid for at their actual length rather than up to the current approximation ratio, this candidate improves on that ratio when the parity repair cost is large. Hence the better of the two candidates always improves on the current approximation ratio, and a constant number of iterations yields a ratio arbitrarily close to that of the TSP algorithm.

\textbf{Adapting the framework to prize collecting.} We retain the interface structure and the iterative scheme and adapt the dynamic program to penalties; an essential change concerns the first candidate. Since in PC-$\Phi$-TSP one must also decide which vertices to visit, it is difficult to compute directly a multigraph satisfying both the $I$-connectivity and the $\mathcal P$-connectivity constraints. Instead, our basic algorithm (Section~\ref{sec:basic}) contracts the prescribed vertices into a single root and applies the $\rho$-approximation algorithm for PCTSP, which decides which vertices to visit. Expanding the returned tour gives an edge multiset that satisfies the $I$-connectivity constraint but may violate the $\mathcal P$-connectivity and parity constraints. We repair $\mathcal P$-connectivity by solving a Steiner forest problem and parity by solving a minimum $T$-join problem, where $T$ is the set of prescribed vertices with the wrong parity; this gives the basic solution. Its objective value is at most $\rho$ times the optimum plus terms measuring the two repair costs (Lemma~\ref{lem:basic-guarantee}), so it handles the case in which both repair costs are small. The two repair steps also produce two laminar families of bounded width (Lemma~\ref{lem:basic-width}), obtained from dual solutions as recalled in Section~\ref{sec:prelim}. We apply the adapted dynamic program (Section~\ref{sec:dp}) to each of them, to handle the cases of a large $\mathcal P$-connectivity repair cost and a large parity repair cost, respectively. In the adapted dynamic program, the endpoints of the guessed edges become prescribed vertices of the local instances and pay no penalty there, but the optimum solution whose edges are guessed visits them, so its penalties still decompose over the local instances. Hence the bound of the dynamic program, in which the guessed edges are paid for at their actual length, extends to the prize-collecting setting (Lemma~\ref{lem:dp-guarantee}).

\textbf{The crossing condition and recursive splitting.} Both the analysis of Traub, Vygen, and Zenklusen and ours bound the candidate computed by the dynamic program under the crossing condition that some optimum solution crosses every set in the laminar family. We prove this bound for PC-$\Phi$-TSP in Lemma~\ref{lem:dual-charging}: under the crossing condition, the candidate improves on the current approximation ratio when the corresponding repair cost is large. For $\Phi$-TSP, this condition holds automatically for the laminar family from the minimum $T$-join problem: $T=Q$, and one can show that every feasible solution crosses every set in this family. For PC-$\Phi$-TSP, however, our basic algorithm contracts the prescribed vertices, so $T$ need not equal $Q$, and an optimum solution need not cross every set in this family. Then Lemma~\ref{lem:dual-charging} no longer applies to this family. To overcome this difficulty, a recursive improvement procedure (Section~\ref{sec:recursive}) adds a fourth type of candidate. A set in the laminar family from the minimum $T$-join problem that an optimum solution does not cross is eligible: splitting the instance along it gives two instances with fewer prescribed vertices whose optimum values sum to at most the optimum value of the original instance (Lemma~\ref{lem:eligible-decomposition}). The procedure tries every eligible set in this family, solves the two resulting instances recursively, and takes the multiunion of the returned edge multisets, so the approximation guarantee for the smaller instances carries over to the original instance. Thus the fourth type of candidate covers the case in which the crossing condition fails, and otherwise Lemma~\ref{lem:dual-charging} applies. Returning the best of these four types of candidates, rather than the two candidates of Traub, Vygen, and Zenklusen, therefore improves the approximation ratio (Lemma~\ref{lem:one-step}), and applying this improvement a constant number of times, following their iterative scheme, proves Theorem~\ref{thm:main}.

\section{Preliminaries}
\label{sec:prelim}

We fix notation used throughout the paper. All graphs are undirected. Unless stated otherwise, $G=(V,E)$ denotes the complete graph of the input PC-$\Phi$-TSP instance with nonnegative edge lengths $\ell:E\to\mathbb R_{\ge0}$ satisfying the triangle inequality and vertex penalties $\pi:V\to\mathbb R_{\ge0}$. For $S\subseteq V$, $G[S]$ denotes the subgraph of $G$ induced by $S$, with edge set $E(S)$. Since $G$ is complete and metric, so is $G[S]$. The set $\delta(S)$ consists of the edges with exactly one endpoint in $S$. We write $\pi(S):=\sum_{v\in S}\pi(v)$. We use $\triangle$ for symmetric difference, i.e., $S_1\triangle S_2=(S_1\setminus S_2)\cup(S_2\setminus S_1)$. For an edge multiset $F$ and an edge set $E'\subseteq E$, let $F\cap E'$ be the submultiset consisting of all copies of edges of $F$ that belong to $E'$, and write $F\subseteq E'$ if $F=F\cap E'$. We use $\uplus$ for multiunion; sums and cardinalities over edge multisets count multiplicities, and intersections and differences of edge multisets are taken with multiplicities in the natural way. Let $V(F)$ be the set of vertices incident to $F$, let $\ell(F):=\sum_{e\in F}\ell(e)$, let $F[S]$ be the submultiset of edges of $F$ with both endpoints in $S$, and let $\operatorname{odd}(F):=\{v:|F\cap \delta(\{v\})|\text{ is odd}\}$. We say that $F$ crosses $S$ if $F\cap\delta(S)\ne\emptyset$. For any interface $\Phi=(I,Q,\mathcal P)$, write $\pi_S^\Phi(F):=\pi(S\setminus(I\cup V(F)))$.

Following Traub, Vygen, and Zenklusen~\cite{TVZ2022}, we treat tours as edge multisets; since tours here need not visit all vertices, we define them as follows. A nonempty edge multiset $F$ is called a tour if $(V(F),F)$ is connected and $\operatorname{odd}(F)=\emptyset$. For distinct $s,t$, a nonempty edge multiset $F$ is called an $s$-$t$ tour if $(V(F),F)$ is connected and $\operatorname{odd}(F)=\{s,t\}$. The empty edge multiset is called the empty tour. By Euler's theorem, such edge multisets can be traversed as closed walks and $s$-$t$ walks, respectively.

We next recall laminar families, the Steiner forest problem, and the minimum $T$-join problem, together with the results used later. A family $\mathcal L\subseteq 2^V$ is laminar if, for any two sets $L_1,L_2\in\mathcal L$, either $L_1\cap L_2=\emptyset$, $L_1\subseteq L_2$, or $L_2\subseteq L_1$. A subfamily of $\mathcal L$ is pairwise disjoint if its sets are mutually disjoint. The width of $\mathcal L$, denoted by $\operatorname{width}(\mathcal L)$, is the maximum size of a pairwise disjoint subfamily of $\mathcal L$; in particular, $\operatorname{width}(\emptyset)=0$.

In a graph $G=(V,E)$ with edge lengths $\ell$, given demand pairs $\mathcal D\subseteq \binom{V}{2}$, the Steiner forest problem asks for an edge set $K\subseteq E$ minimizing $\ell(K)$ such that $u$ and $v$ lie in the same connected component of $(V,K)$ for every $\{u,v\}\in\mathcal D$; given a set $T\subseteq V$ of even cardinality, the minimum $T$-join problem asks for an edge set $J\subseteq E$ minimizing $\ell(J)$ subject to $\operatorname{odd}(J)=T$. These two problems admit linear programming relaxations and duals of the following common form.
\[
\begin{array}{@{}rll@{\qquad\qquad}rll@{}}
(\mathcal S\text{-}\mathrm{LP})\quad
\min & \displaystyle\sum_{e\in E}\ell(e)x_e
&
&
(\mathcal S\text{-}\mathrm{D})\quad
\max & \displaystyle\sum_{S\in\mathcal S} y_S
&
\\[1ex]
\mathrm{s.t.} & \displaystyle\sum_{e\in\delta(S)}x_e\ge 1 & \forall S\in\mathcal S,
&
\mathrm{s.t.} & \displaystyle\sum_{S\in\mathcal S: e\in\delta(S)} y_S\le \ell(e) & \forall e\in E,
\\
& x_e\ge0 & \forall e\in E,
&
& y_S\ge0 & \forall S\in\mathcal S.
\end{array}
\]
For a Steiner forest instance with demand pairs $\mathcal D$, let $\mathcal S_{\mathcal D}:=\{S\subseteq V: |S\cap\{u,v\}|=1 \text{ for some } \{u,v\}\in\mathcal D\}$. Taking $\mathcal S=\mathcal S_{\mathcal D}$ gives the linear programming relaxation for the Steiner forest instance and its dual, denoted by $(\mathrm{SF}\text{-}\mathrm{LP})$ and $(\mathrm{SF}\text{-}\mathrm{D})$. For a minimum $T$-join instance, let $\mathcal S_T:=\{S\subseteq V: |S\cap T|\text{ is odd}\}$. Taking $\mathcal S=\mathcal S_T$ gives the linear programming relaxation for the minimum $T$-join problem and its dual, denoted by $(\mathrm{Join}\text{-}\mathrm{LP})$ and $(\mathrm{Join}\text{-}\mathrm{D})$. We use the following results on the two problems. They bound the two repair costs of the basic algorithm (Section~\ref{sec:basic}).

\begin{lemma}
\label{lem:sf-gw}
For every Steiner forest instance, the primal-dual algorithm of Goemans and Williamson~\cite{GW1995} computes in polynomial time a feasible solution $K$ to the instance, a feasible solution $y^{\mathcal D}$ to $(\mathrm{SF}\text{-}\mathrm{D})$, and a laminar family $\mathcal L_{\mathcal D}:=\{S\in\mathcal S_{\mathcal D}:y^{\mathcal D}_S>0\}$ such that $\ell(K)\le 2\sum_{S\in\mathcal L_{\mathcal D}} y^{\mathcal D}_S$. By weak duality, $\sum_{S\in\mathcal L_{\mathcal D}}y^{\mathcal D}_S$ is a lower bound on the optimum value of the Steiner forest instance.
\end{lemma}

\begin{lemma}
\label{lem:tjoin-laminar}
For every minimum $T$-join instance, one can compute in polynomial time an optimum solution $J\subseteq E$ and an optimum solution $y^T$ to $(\mathrm{Join}\text{-}\mathrm{D})$ such that the family $\mathcal L_T:=\{S\in\mathcal S_T:y^T_S>0\}$ is laminar. Moreover, $(\mathrm{Join}\text{-}\mathrm{LP})$ has optimum value $\ell(J)$, and $\sum_{S\in\mathcal L_T} y^T_S=\ell(J)$.
\end{lemma}

The fact that $(\mathrm{Join}\text{-}\mathrm{LP})$ has optimum value $\ell(J)$ follows from the integrality of the $T$-join polyhedron; see, e.g.,~\cite{Schrijver2003}. A dual optimum whose positive variables form a laminar family can be obtained by an uncrossing procedure~\cite{Karzanov1996}; see also~\cite[Lemma~18]{TVZ2022}. When $T=\emptyset$, we take $J=\emptyset$, $y^T=0$, and $\mathcal L_T=\emptyset$.

\begin{remark}
\label{rem:standard-components}
The pruning step in the algorithm of~\cite{GW1995} for the Steiner forest problem leaves every nonempty component of $(V(K),K)$ with an endpoint of some demand pair. Similarly, whenever we compute a minimum $T$-join, we take a returned solution $J$ such that every connected component of $(V(J),J)$ contains a vertex of $T$; this is possible since a component without a vertex of $T$ has only even degrees and can be removed without increasing the length. We use these properties in the proof of Lemma~\ref{lem:basic-guarantee} to verify the $I$-connectivity constraint.
\end{remark}

\section{The basic algorithm}
\label{sec:basic}

We first design a basic algorithm for PC-$\Phi$-TSP. Given a PC-$\Phi$-TSP instance with $\Phi=(I,Q,\mathcal P)$, if $I=\emptyset$, then $\Phi=(\emptyset,\emptyset,\emptyset)$, and the empty edge multiset is the only feasible solution.

If $I\ne\emptyset$, we first construct an edge multiset $H$ satisfying the $I$-connectivity constraint. \textbf{(H1)} We construct an auxiliary PCTSP instance. Let $r_I$ be a new root representing $I$. For each $v\in V\setminus I$, choose $\sigma(v)\in I$ such that $\ell(\sigma(v)v)=\min_{w\in I}\ell(wv)$. Let $G_I$ be the complete graph on $(V\setminus I)\cup\{r_I\}$ in which each edge with both endpoints in $V\setminus I$ has its original length, and each edge $r_Iv$ has length $\ell(\sigma(v)v)$. For $u,v\in (V\setminus I)\cup\{r_I\}$, let $d_I(u,v)$ be the length of a shortest $u$-$v$ path in $G_I$. The complete graph on $(V\setminus I)\cup\{r_I\}$ with edge lengths $d_I$, root $r_I$, penalties $\pi(v)$ for $v\in V\setminus I$, and penalty $0$ for $r_I$ defines the auxiliary PCTSP instance. \textbf{(H2)} We run a polynomial-time $\rho$-approximation algorithm for PCTSP on this auxiliary instance. \textbf{(H3)} We expand the solution to the auxiliary instance back to the given graph. For each edge $e=uv$ in the solution to the auxiliary instance, replace $e$ by a shortest $u$-$v$ path $P_e$ in $G_I$ and map $P_e$ back to $G$ edge by edge: an edge with both endpoints in $V\setminus I$ remains unchanged, while an edge $r_Iv$ is replaced by the edge $\sigma(v)v$ of $G$. Thus each path $P_e$ becomes a collection of paths in $G$, split at the visits to $r_I$. Let $H$ be the multiunion, over all edges $e$ in the solution to the auxiliary instance, of the edge sets of all paths in the collection obtained from $P_e$.

Since the solution to the auxiliary instance is a possibly empty Eulerian tour rooted at $r_I$, the expansion in $(\mathrm H3)$ implies that, first, every connected component of $(I\cup V(H),H)$ contains a vertex of $I$, so $H$ satisfies the $I$-connectivity constraint; second, degree parity is preserved at all vertices of $V\setminus I$, so $\operatorname{odd}(H)\subseteq I$.

We repair the parity constraint and the $\mathcal P$-connectivity constraint independently. Let $T:=\operatorname{odd}(H)\triangle Q$. Since $Q,\operatorname{odd}(H)\subseteq I$ and both have even cardinality, $T$ is an even-cardinality subset of $I$. We compute a minimum $T$-join $J$; then $\operatorname{odd}(H\uplus J)=\operatorname{odd}(H)\triangle T=Q$, so parity is restored. To restore the $\mathcal P$-connectivity constraint, for each set $C\in\mathcal P$, choose a representative $c_C\in C$ and define $\mathcal D:=\{\{c_C,v\}: C\in\mathcal P,\ v\in C\setminus\{c_C\}\}$. We compute a Steiner forest $K$ for the demand pairs in $\mathcal D$. Then $K$ connects all vertices in each set $C\in\mathcal P$, and adding two copies of $K$ does not change parity. The final edge multiset is $F_{\mathrm{bas}}:=H\uplus J\uplus K\uplus K$. The complete procedure is summarized in Algorithm~\ref{alg:basic} and illustrated in Figure~\ref{fig:basic_algorithm_example}.

\begin{algorithm}[!ht]\small
\caption{\textsc{BasicApprox}$(G,\ell,\pi,\Phi)$}
\label{alg:basic}
\LinesNumbered
\noindent\textbf{Input:} a PC-$\Phi$-TSP instance with $\Phi=(I,Q,\mathcal P)$.\\
\noindent\textbf{Output:} a feasible solution to the input instance and two laminar families.
\BlankLine
\nl \If{$I=\emptyset$}{
    \Return $F_{\mathrm{bas}}:=\emptyset$, $\mathcal L_T:=\emptyset$, and $\mathcal L_{\mathcal D}:=\emptyset$\;
}
\nl Construct the edge multiset $H$ by $(\mathrm H1)$--$(\mathrm H3)$\;
\nl $T\gets \operatorname{odd}(H)\triangle Q$\;
\nl Using Lemma~\ref{lem:tjoin-laminar}, compute, with respect to $G$, $\ell$, and $T$, a minimum $T$-join $J$ and an optimum solution $y^T$ of $(\mathrm{Join}\text{-}\mathrm{D})$ such that $\mathcal L_T:=\{S\in\mathcal S_T:y^T_S>0\}$ is laminar\;
\nl For each set $C\in\mathcal P$, choose a representative $c_C\in C$, and let $\mathcal D:=\{\{c_C,v\}: C\in\mathcal P,\ v\in C\setminus\{c_C\}\}$\;
\nl Using Lemma~\ref{lem:sf-gw}, compute, with respect to $G$, $\ell$, and $\mathcal D$, a Steiner forest $K$, a feasible solution $y^{\mathcal D}$ of $(\mathrm{SF}\text{-}\mathrm{D})$, and the laminar family $\mathcal L_{\mathcal D}:=\{S\in\mathcal S_{\mathcal D}:y^{\mathcal D}_S>0\}$\;
\nl $F_{\mathrm{bas}}\gets H\uplus J\uplus K\uplus K$\;
\nl \Return $F_{\mathrm{bas}}, \mathcal L_T, \mathcal L_{\mathcal D}$.
\end{algorithm}

\begin{figure}[htpb]\small 
\centering
\includegraphics[height=0.25\textheight,keepaspectratio]{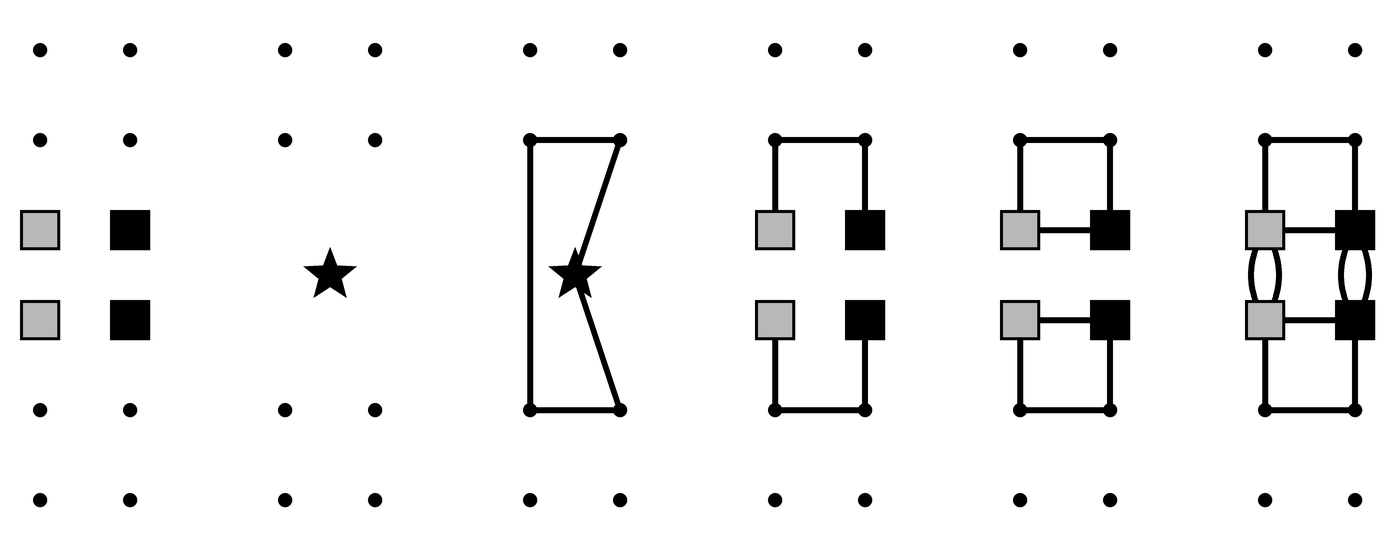}
\caption{An example of the basic algorithm. From left to right: the input PC-$\Phi$-TSP instance; the auxiliary PCTSP instance; a solution to the auxiliary PCTSP instance returned by the $\rho$-approximation algorithm for PCTSP; the expansion of this solution back to the given graph; the repair of the parity constraint; and the repair of the $\mathcal P$-connectivity constraint. The star denotes the root $r_I$ of the auxiliary PCTSP instance.}
\label{fig:basic_algorithm_example}
\end{figure}

Let $F^*$ be an optimum solution to the input PC-$\Phi$-TSP instance and set $\mathrm{OPT}:=\ell(F^*)+\pi_V^\Phi(F^*)$. We now analyze the approximation guarantee of Algorithm~\ref{alg:basic}.

\begin{lemma}
\label{lem:basic-guarantee}
Let $(G,\ell,\pi,\Phi)$ be a PC-$\Phi$-TSP instance with $\Phi=(I,Q,\mathcal P)$ and optimum value $\mathrm{OPT}$. Algorithm~\ref{alg:basic} runs in polynomial time and returns a feasible solution $F_{\mathrm{bas}}$ to this instance. If $I\ne\emptyset$, let $H$ be the edge multiset constructed in $(\mathrm H1)$--$(\mathrm H3)$, let $J$ be the computed minimum $T$-join for $T=\operatorname{odd}(H)\triangle Q$, and let $y^{\mathcal D}$ be the computed feasible solution to $(\mathrm{SF}\text{-}\mathrm{D})$. Then
\[
    \ell(F_{\mathrm{bas}})+\pi_V^\Phi(F_{\mathrm{bas}})
    \le
    \rho\,\mathrm{OPT}
    +\ell(J)
    +4\sum_{S\in\mathcal L_{\mathcal D}} y_S^{\mathcal D}.
\]
Consequently, Algorithm~\ref{alg:basic} is a $(2\rho+5)$-approximation algorithm.
\end{lemma}

\begin{proof}
If $I=\emptyset$, the algorithm returns the empty edge multiset, which is the only feasible solution. Hence assume $I\ne\emptyset$. We now check feasibility. By construction, $F_{\mathrm{bas}}=H\uplus J\uplus K\uplus K$. Since $J$ is a $T$-join and $K$ is added twice, $\operatorname{odd}(F_{\mathrm{bas}})=\operatorname{odd}(H)\triangle T=Q$. The Steiner forest $K$ connects all vertices in each set $C\in\mathcal P$, so the $\mathcal P$-connectivity constraint holds. By Remark~\ref{rem:standard-components}, every connected component of $(V(J),J)$ contains a vertex of $T\subseteq I$, and every nonempty connected component of $(V(K),K)$ contains an endpoint of a demand pair in $\mathcal D$, hence a vertex of $I$. Together with the $I$-connectivity of $H$, this implies that every connected component of $(I\cup V(F_{\mathrm{bas}}),F_{\mathrm{bas}})$ contains a vertex of $I$. Thus $F_{\mathrm{bas}}$ is feasible.

We next analyze the objective value. The optimum solution $F^*$ induces a feasible solution to the auxiliary PCTSP instance from $(\mathrm H1)$ by contracting $I$ to $r_I$: edges with both endpoints in $V\setminus I$ are kept, edges with one endpoint in $I$ are incident to $r_I$, and edges with both endpoints in $I$ can be ignored. After contracting $I$ to $r_I$, the resulting multigraph is connected through $r_I$ and all non-root vertices have even degree, and hence $r_I$ also has even degree. By the definition of $G_I$ and $d_I$, the induced solution has length at most $\ell(F^*)$, and its penalty is at most $\pi_V^\Phi(F^*)$. Hence the PCTSP auxiliary instance has optimum objective value at most $\ell(F^*)+\pi_V^\Phi(F^*)$. The polynomial-time $\rho$-approximation algorithm for PCTSP used in $(\mathrm H2)$ returns a solution to this auxiliary instance with objective value at most $\rho(\ell(F^*)+\pi_V^\Phi(F^*))$. Expanding this solution back to $G$ as in $(\mathrm H3)$ does not increase its length and does not lose any vertex of $V\setminus I$ visited by the solution to the auxiliary instance. Therefore $\ell(H)+\pi_V^\Phi(H)\le \rho(\ell(F^*)+\pi_V^\Phi(F^*))=\rho\mathrm{OPT}$.

We now bound the repair costs. Since $\operatorname{odd}(H\uplus F^*)=\operatorname{odd}(H)\triangle Q=T$, the edge multiset $H\uplus F^*$ contains a $T$-join of length at most $\ell(H)+\ell(F^*)$. Since $J$ is a minimum $T$-join, $\ell(J)\le \ell(H)+\ell(F^*)\le(\rho+1)\mathrm{OPT}$.

For the $\mathcal P$-connectivity repair, since $F^*$ connects every set $C\in\mathcal P$, it contains a Steiner forest feasible for the instance with demand pairs $\mathcal D$, of length at most $\ell(F^*)$. By Lemma~\ref{lem:sf-gw}, we have $\ell(K)\le 2\sum_{S\in\mathcal L_{\mathcal D}}y^{\mathcal D}_S$. Moreover, $\sum_{S\in\mathcal L_{\mathcal D}}y^{\mathcal D}_S$ is a lower bound on the optimum value of this Steiner forest instance. Hence $\sum_{S\in\mathcal L_{\mathcal D}}y^{\mathcal D}_S\le \ell(F^*)$, and therefore $\ell(K)\le 2\ell(F^*)\le2\mathrm{OPT}$.

Finally, since $I\cup V(H)\subseteq I\cup V(F_{\mathrm{bas}})$, we have $\pi_V^\Phi(F_{\mathrm{bas}})\le \pi_V^\Phi(H)$. Therefore $\ell(F_{\mathrm{bas}})+\pi_V^\Phi(F_{\mathrm{bas}})\le \ell(H)+\pi_V^\Phi(H)+\ell(J)+2\ell(K)$. Combining the bounds above gives
\[
    \ell(F_{\mathrm{bas}})+\pi_V^\Phi(F_{\mathrm{bas}})
    \le \rho\,\mathrm{OPT}
    +\ell(J)
    +4\sum_{S\in\mathcal L_{\mathcal D}} y_S^{\mathcal D}\le(2\rho+5)\mathrm{OPT}.
\]

The running time is polynomial by the construction of $H$, the polynomial running time of the $\rho$-approximation algorithm for PCTSP used in $(\mathrm H2)$, and Lemmas~\ref{lem:sf-gw} and~\ref{lem:tjoin-laminar}.
\end{proof}

\begin{lemma}
\label{lem:basic-width}
For an input instance with interface $\Phi=(I,Q,\mathcal P)$, the laminar families $\mathcal L_{\mathcal D}$ and $\mathcal L_T$ computed in Algorithm~\ref{alg:basic} satisfy $\operatorname{width}(\mathcal L_{\mathcal D})\le |I|$ and $\operatorname{width}(\mathcal L_T)\le |I|$.
\end{lemma}

\begin{proof}
If $I=\emptyset$, the claim is immediate. Assume $I\ne\emptyset$. By Lemma~\ref{lem:sf-gw}, every set of $\mathcal L_{\mathcal D}$ belongs to $\mathcal S_{\mathcal D}$, and hence separates a demand pair in $\mathcal D$. Both endpoints of every demand pair in $\mathcal D$ belong to $I$, so every set of $\mathcal L_{\mathcal D}$ contains a vertex of $I$. Therefore every pairwise disjoint subfamily of $\mathcal L_{\mathcal D}$ has size at most $|I|$, and $\operatorname{width}(\mathcal L_{\mathcal D})\le |I|$. By Lemma~\ref{lem:tjoin-laminar}, every set of $\mathcal L_T$ belongs to $\mathcal S_T$, and hence contains an odd number of vertices of $T$. Thus every set of $\mathcal L_T$ contains a vertex of $T$, so every pairwise disjoint subfamily of $\mathcal L_T$ has size at most $|T|$. Hence $\operatorname{width}(\mathcal L_T)\le |T|$. Finally, $T=\operatorname{odd}(H)\triangle Q \subseteq I$, and so $|T|\le |I|$.
\end{proof}

\section{The dynamic program}
\label{sec:dp}

The second algorithm is a dynamic program that, given a PC-$\Phi$-TSP instance with interface $\Phi=(I,Q,\mathcal P)$, a fixed positive integer $h$, an explicitly given laminar family $\mathcal L$, and a $\beta$-approximation algorithm $\mathcal A_\beta$ for PC-$\Phi$-TSP that runs in polynomial time on instances whose number of prescribed vertices is bounded by a fixed constant, returns a feasible solution in polynomial time when $|I|$ and $\operatorname{width}(\mathcal L)$ are bounded by constants. As in~\cite{TVZ2022}, the dynamic program guesses edges in an optimum solution; larger $h$ allows more such edges to be guessed, but increases the running time. In Section~\ref{sec:recursive}, we apply the algorithm to $\mathcal L_T$ and $\mathcal L_{\mathcal D}$ computed by Algorithm~\ref{alg:basic}.

The dynamic program constructs and solves instances on certain subgraphs of the given graph. To avoid notational confusion, we call them local PC-$\Psi$-TSP instances, where $\Psi$ denotes the local interface. For $S\subseteq V$, a local interface on $S$ is a triple $\Psi=(I_\Psi,Q_\Psi,\mathcal P_\Psi)$, where $I_\Psi\subseteq S$, $Q_\Psi\subseteq I_\Psi$ has even cardinality, and $\mathcal P_\Psi$ is a partition of $I_\Psi$. The corresponding local PC-$\Psi$-TSP instance is the instance on $G[S]$ with local interface $\Psi$. A local feasible solution for this instance is an edge multiset $F\subseteq E(S)$ such that $(I_\Psi\cup V(F),F)$ satisfies the parity, $I_\Psi$-connectivity, and $\mathcal P_\Psi$-connectivity constraints for $\Psi$. Its local objective value is $\ell(F)+\pi_S^\Psi(F)$.

We now describe the dynamic program. Fix a PC-$\Phi$-TSP instance with $\Phi=(I,Q,\mathcal P)$, a laminar family $\mathcal L$, and a fixed positive integer $h$. Let $w:=\operatorname{width}(\mathcal L)$ and set $B:=|I|+h(w+1)$. For every set $S\in\mathcal L\cup\{V\}$, the dynamic program enumerates all local interfaces $\Psi$ on $S$ with $|I_\Psi|\le B$, including the empty interface $(\emptyset,\emptyset,\emptyset)$. For each such interface $\Psi$, it computes a feasible solution $Z[S,\Psi]$ to the corresponding local PC-$\Psi$-TSP instance. The sets $S\in\mathcal L\cup\{V\}$ are processed in increasing order of $|S|$. At the end, the algorithm returns $Z[V,\Phi]$, which is a feasible solution to the input PC-$\Phi$-TSP instance. When $|I|$, $h$, and $w$ are bounded by constants, the number of enumerated interfaces with $|I_\Psi|\le B$ is polynomial in $|V|$.

When a set $S\in\mathcal L\cup\{V\}$ is processed, the order above ensures that all local solutions $Z[S',\Psi']$ with $S'\in\mathcal L$, $S'\subsetneq S$, and $|I_{\Psi'}|\le B$ have already been computed. For each enumerated interface $\Psi_S=(I_S,Q_S,\mathcal P_S)$ on $S$, the algorithm enumerates all child decompositions of $S$. A child decomposition is a possibly empty pairwise disjoint collection $S_1,\ldots,S_p$ of sets of $\mathcal L$ with $S_i\subsetneq S$ for every $i\in\{1,\ldots,p\}$. If $S$ contains no set $S'\in\mathcal L$ with $S'\subsetneq S$, then the empty child decomposition is the only child decomposition. Let $S_0:=S\setminus\bigcup_{i=1}^p S_i$. We call $S_0$ the residual part and call $S_0,S_1,\ldots,S_p$ the parts of this child decomposition. Since every child decomposition contains at most $w$ sets, the number of child decompositions for a fixed $S$ is at most $(|\mathcal L|+1)^w$, which is polynomial for fixed $w$ because $\mathcal L$ is laminar and $|\mathcal L|=O(|V|)$.

For the chosen child decomposition, the algorithm enumerates all boundary multisets. A boundary multiset is an edge multiset $X\subseteq E(S)$ such that each edge of $X$ has endpoints in two distinct parts among $S_0,S_1,\ldots,S_p$, and $|X\cap\delta(S_i)|\le h$ for every $i\in\{1,\ldots,p\}$. Since $p\le w$, we have $|X|\le hw$, so the number of possible boundary multisets is polynomial in $|V|$ when $h$ and $w$ are fixed.

For each boundary multiset $X$ and each $i\in\{0,\ldots,p\}$, set $I_i:=(I_S\cup V(X))\cap S_i$ and $Q_i:=I_i\cap(Q_S\triangle\operatorname{odd}(X))$. If there exists $i\in\{0,\ldots,p\}$ such that $|I_i|>B$ or $|Q_i|$ is odd, this choice of $X$ is rejected and the algorithm proceeds to the next boundary multiset. Otherwise, the algorithm enumerates all partitions $\mathcal P_i$ of $I_i$ for which $\Psi_i=(I_i,Q_i,\mathcal P_i)$ is an interface on $S_i$ for every $i\in\{0,\ldots,p\}$. For $i\in\{1,\ldots,p\}$, let $F_i:=Z[S_i,\Psi_i]$, which has been computed since $|I_i|\le B$ and $S_i\subsetneq S$. The residual part $S_0$ is handled by applying $\mathcal A_\beta$ to the local PC-$\Psi_0$-TSP instance on $G[S_0]$; let $F_0$ be the returned edge multiset. We compute $F':=X\uplus F_0\uplus F_1\uplus\cdots\uplus F_p$. If $F'$ is feasible for the local PC-$\Psi_S$-TSP instance on $G[S]$, we keep $F'$ as a candidate.

Among all candidates constructed over all enumerated choices for the local PC-$\Psi_S$-TSP instance on $G[S]$, the dynamic program chooses one with minimum local objective value and denotes it by $Z[S,\Psi_S]$. The empty child decomposition always produces a candidate: then $S_0=S$ and $X=\emptyset$, and choosing $\mathcal P_0=\mathcal P_S$ gives $\Psi_0=\Psi_S$, so $\mathcal A_\beta$ is applied to the local PC-$\Psi_S$-TSP instance on $G[S]$. Hence, for every enumerated interface $\Psi_S$ on $S$ with $|I_{\Psi_S}|\le B$, the solution $Z[S,\Psi_S]$ is defined after $S$ is processed. Since $Z[V,\Phi]$ is defined when the algorithm terminates, the dynamic program returns a feasible solution to the input instance.

The complete procedure is summarized in Algorithm~\ref{alg:dp} and illustrated in Figure~\ref{fig:dp_example}.

\begin{algorithm}[!ht]\small
\caption{\textsc{LaminarDP}$(G,\ell,\pi,\Phi,\mathcal L,h,\mathcal A_\beta)$}
\label{alg:dp}
\LinesNumbered
\noindent\textbf{Input:} a PC-$\Phi$-TSP instance with $\Phi=(I,Q,\mathcal P)$, a laminar family $\mathcal L$, a fixed positive integer $h$, and a $\beta$-approximation algorithm $\mathcal A_\beta$ for PC-$\Phi$-TSP that runs in polynomial time on instances whose number of prescribed vertices is bounded by a fixed constant.\\
\noindent\textbf{Output:} a feasible solution to the input instance.
\BlankLine
\nl $w\gets\operatorname{width}(\mathcal L)$, $B\gets |I|+h(w+1)$\;
\nl \For{each $S\in\mathcal L\cup\{V\}$, in increasing order of cardinality}{
  \For{each interface $\Psi_S=(I_S,Q_S,\mathcal P_S)$ on $S$ with $|I_S|\le B$}{
    $Z[S,\Psi_S]\gets\mathrm{undefined}$\;
    \For{each child decomposition $S_1,\ldots,S_p$ of $S$}{
      $S_0\gets S\setminus\bigcup_{i=1}^p S_i$\;
      \For{each boundary multiset $X$ for this decomposition}{
        \For{$i=0,\ldots,p$}{
          $I_i\gets (I_S\cup V(X))\cap S_i$, $Q_i\gets I_i\cap(Q_S\triangle\operatorname{odd}(X))$\;
        }
        \If{there exists $i\in\{0,\ldots,p\}$ such that $|I_i|>B$ or $|Q_i|$ is odd}{
          reject this choice of $X$ and \textbf{continue} with the next boundary multiset\;
        }
        \For{each choice of partitions $\mathcal P_0,\ldots,\mathcal P_p$ such that $\Psi_i=(I_i,Q_i,\mathcal P_i)$ is an interface on $S_i$ for every $i\in\{0,\ldots,p\}$}{
          $F_i\gets Z[S_i,\Psi_i]$ for $i=1,\ldots,p$\;
          run $\mathcal A_\beta$ on the corresponding local PC-$\Psi_0$-TSP instance on $G[S_0]$ and let $F_0$ be the returned edge multiset (when $S_0=\emptyset$, this means $F_0=\emptyset$)\;
          $F'\gets X\uplus F_0\uplus F_1\uplus\cdots\uplus F_p$\;
          \If{$F'$ is feasible for the local PC-$\Psi_S$-TSP instance on $G[S]$, and either $Z[S,\Psi_S]$ is undefined or $F'$ has smaller local objective value than $Z[S,\Psi_S]$}{
            $Z[S,\Psi_S]\gets F'$\;
          }
        }
      }
    }
  }
}
\nl \Return $Z[V,\Phi]$\;
\end{algorithm}

\begin{figure}[htpb]\small 
\centering
\includegraphics[height=0.16\textheight,keepaspectratio]{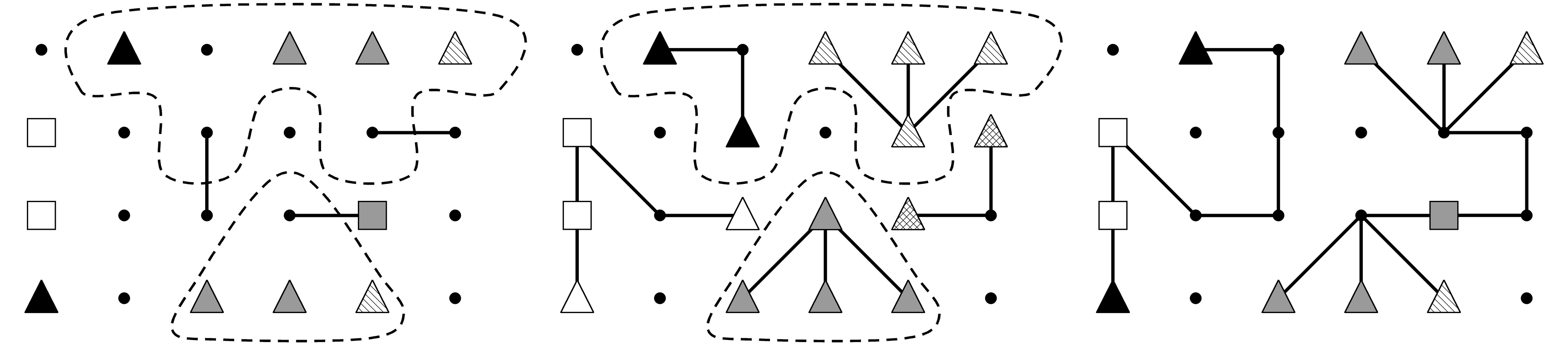}
\caption{An example of the dynamic program. From left to right: a chosen child decomposition with a boundary multiset; chosen local interfaces and edge multisets computed inside the parts of the child decomposition; and the multiunion of these edge multisets with the boundary multiset, forming a candidate.}
\label{fig:dp_example}
\end{figure}

Consider Algorithm~\ref{alg:dp} applied to $\mathcal L_T$ or $\mathcal L_{\mathcal D}$ on an instance with $|I|\le k$. Its running time is bounded by the running time of $\mathcal A_\beta$ on instances with at most $B$ prescribed vertices, multiplied by a polynomial in $|V|$ whose exponent depends only on $k$ and $h$. Indeed, Lemma~\ref{lem:basic-width} gives $w=\operatorname{width}(\mathcal L)\le |I|\le k$ in these applications, so $B=|I|+h(w+1)\le k+h(k+1)$. The numbers of local interfaces, child decompositions, boundary multisets, and partition choices are polynomial in $|V|$ with exponents depending only on $B$, $h$, and $w$, so Algorithm~\ref{alg:dp} considers polynomially many $F'$, each requiring one call to $\mathcal A_\beta$ on a local instance with at most $B$ prescribed vertices and polynomial additional work.

Fix an optimum solution $F^*$, and set $\mathrm{OPT}:=\ell(F^*)+\pi_V^\Phi(F^*)$. For $S\subseteq V$, the interface $\Psi_S^*:=(I_S^*,Q_S^*,\mathcal P_S^*)$ induced by $F^*$ on $S$ is defined as follows. Let $I_S^*:=(I\cup V(F^*\cap\delta(S)))\cap S$ and $Q_S^*:=I_S^*\cap\operatorname{odd}(F^*[S])$. If $S=V$, set $\mathcal P_S^*:=\mathcal P$; otherwise, let $\mathcal P_S^*:=\{D\cap I_S^*: D \text{ is a connected component of } (I_S^*\cup V(F^*[S]),F^*[S]) \text{ and } D\cap I_S^*\ne\emptyset\}$. Every vertex of $S\setminus I_S^*$ is not in $Q$ and has the same degree in $F^*[S]$ as in $F^*$, so this degree is even. Hence $\operatorname{odd}(F^*[S])\subseteq I_S^*$, and $Q_S^*=\operatorname{odd}(F^*[S])$ has even cardinality. In particular, $\Psi_V^*=\Phi$. Moreover, $F^*[S]$ is feasible for the local PC-$\Psi_S^*$-TSP instance on $G[S]$: the parity constraint holds since $Q_S^*=\operatorname{odd}(F^*[S])$, the $\mathcal P_S^*$-connectivity constraint holds by definition, and every connected component of $(I_S^*\cup V(F^*[S]),F^*[S])$ either contains an endpoint of an edge of $F^*\cap\delta(S)$ or is a connected component of $(I\cup V(F^*),F^*)$, and in both cases it contains a vertex of $I_S^*$.

Define $\mathcal L(F^*,h):=\{L\in\mathcal L:1\le |F^*\cap\delta(L)|\le h\}$ and $F^*(\mathcal L,h):=F^*\cap\bigcup_{L\in\mathcal L(F^*,h)}\delta(L)$. Thus $\mathcal L(F^*,h)$ consists of the sets in $\mathcal L$ crossed by $F^*$ between one and $h$ times, and $F^*(\mathcal L,h)$ consists of the edges of $F^*$ crossing at least one such set. The edges in $F^*(\mathcal L,h)$ are the edges Algorithm~\ref{alg:dp} tries to guess. As $h$ increases, $F^*(\mathcal L,h)$ may contain more edges of $F^*$. Figure~\ref{fig:laminar_example} illustrates $\mathcal L(F^*,h)$ and $F^*(\mathcal L,h)$.

\begin{figure}[htpb]\small 
\centering
\includegraphics[width=\linewidth,height=0.4\textwidth,keepaspectratio]{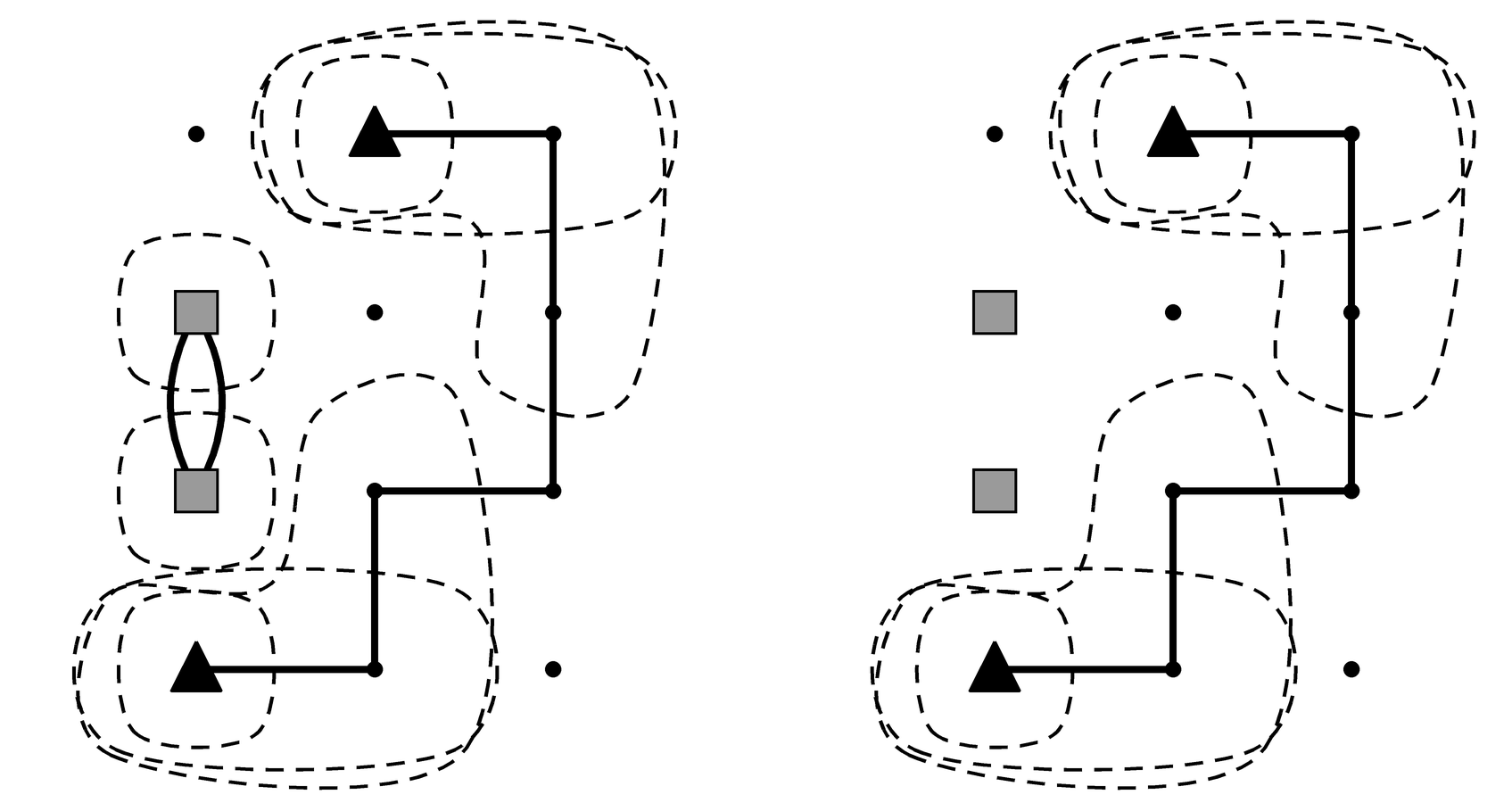}
\caption{An illustration of $\mathcal L(F^*,h)$ and $F^*(\mathcal L,h)$. The left panel shows an optimum solution $F^*$ and the sets of a laminar family $\mathcal L$. The right panel shows the sets in $\mathcal L(F^*,h)$ and the edge multiset $F^*(\mathcal L,h)$ for $h=1$.}
\label{fig:laminar_example}
\end{figure}

\begin{lemma}
\label{lem:dp-guarantee}
Let $\mathcal L$ be a laminar family and $h$ a positive integer, and let $\mathcal A_\beta$ be a $\beta$-approximation algorithm for PC-$\Phi$-TSP. Then the output $Z[V,\Phi]$ of Algorithm~\ref{alg:dp} on input $(G,\ell,\pi,\Phi,\mathcal L,h,\mathcal A_\beta)$ satisfies
\[
\ell(Z[V,\Phi])+\pi_V^\Phi(Z[V,\Phi])
\le
\beta\,\mathrm{OPT}-(\beta-1)\ell\bigl(F^*(\mathcal L,h)\bigr).
\]
\end{lemma}

\begin{proof}

The idea is to follow $F^*$: for a set $L\in\mathcal L(F^*,h)\cup\{V\}$, the dynamic program can guess the child decomposition of $L$ formed by the maximal sets of $\mathcal L(F^*,h)$ strictly contained in $L$, the boundary multiset consisting of the edges of $F^*[L]$ that cross these sets, and the local interfaces induced by $F^*$. The candidate assembled from these guesses pays for the guessed edges at their actual length, and we bound its cost by comparing it, on each local instance, with the restriction of $F^*$ to that instance.

If $I=\emptyset$, Algorithm~\ref{alg:dp} returns the empty edge multiset, and the claimed bound is immediate.

Consider the case $I\ne\emptyset$. We prove by induction over the sets in $\mathcal L(F^*,h)\cup\{V\}$, in increasing order of cardinality, that after a set $L$ has been processed,
\[
    \ell(Z[L,\Psi_L^*])+\pi_L^{\Psi_L^*}(Z[L,\Psi_L^*])
    \le
    \beta\bigl(\ell(F^*[L]\setminus F^*(\mathcal L,h))
        +\pi_L^{\Psi_L^*}(F^*[L])\bigr)
    +\ell(F^*[L]\cap F^*(\mathcal L,h)).
\]
For every $L\in\mathcal L(F^*,h)\cup\{V\}$, we have $|I_L^*|\le |I|+h\le B$, so $\Psi_L^*$ is among the interfaces enumerated for $L$.

First consider a set $L\in\mathcal L(F^*,h)\cup\{V\}$ that contains no set $L'\in\mathcal L(F^*,h)$ with $L'\subsetneq L$. For the empty child decomposition, the algorithm applies $\mathcal A_\beta$ to the local PC-$\Psi_L^*$-TSP instance on $G[L]$. Since $F^*[L]$ is feasible for this local instance, this produces a candidate of local objective value at most $\beta\bigl(\ell(F^*[L])+\pi_L^{\Psi_L^*}(F^*[L])\bigr)$. Hence $Z[L,\Psi_L^*]$ has local objective value at most this quantity. Moreover, by laminarity and the choice of $L$, we have $F^*[L]\cap F^*(\mathcal L,h)=\emptyset$. Thus the claimed bound holds for $L$.

Now fix a set $L\in\mathcal L(F^*,h)\cup\{V\}$ that contains at least one set $L'\in\mathcal L(F^*,h)$ with $L'\subsetneq L$, and assume that the claim has already been proved for all smaller sets in this family. Let $L_1,\ldots,L_q$ be the inclusionwise maximal sets among all sets $S\in\mathcal L(F^*,h)$ with $S\subsetneq L$. Since $\mathcal L(F^*,h)$ is laminar, the sets $L_1,\ldots,L_q$ are pairwise disjoint, and hence form one of the child decompositions considered by Algorithm~\ref{alg:dp}. Put $L_0:=L\setminus\bigcup_{i=1}^q L_i$, the residual part of this child decomposition, and $X^L:=F^*[L]\setminus\bigl(F^*[L_0]\uplus F^*[L_1]\uplus\cdots\uplus F^*[L_q]\bigr)$. Each edge of $X^L$ has endpoints in different parts and crosses at least one of the sets $L_1,\ldots,L_q$, so $X^L\subseteq F^*[L]\cap F^*(\mathcal L,h)$. Moreover, $|X^L\cap\delta(L_i)|\le |F^*\cap\delta(L_i)|\le h$ for every $i\in\{1,\ldots,q\}$, and hence $|X^L|\le hq$. Thus $X^L$ is one of the boundary multisets considered for this child decomposition.

Fix this child decomposition $L_1,\ldots,L_q$ and the boundary multiset $X^L$. For each $i\in\{0,\ldots,q\}$, the algorithm produces the interface induced by $F^*$ on $L_i$, that is,
\[
(I_L^*\cup V(X^L))\cap L_i=I_{L_i}^*
\qquad\text{and}\qquad
I_{L_i}^*\cap\bigl(Q_L^*\triangle\operatorname{odd}(X^L)\bigr)=Q_{L_i}^*.
\]
Indeed, recall that $I_S^*$ consists of the vertices of $S$ that belong to $I$ or are incident to an edge of $F^*$ leaving $S$, and that $Q_S^*=\operatorname{odd}(F^*[S])\subseteq I_S^*$. An edge of $F^*$ leaving $L_i$ either leaves $L$ or joins $L_i$ to another part, in which case it belongs to $X^L$; this gives the first identity. Similarly, at every vertex of $L_i$, the edges of $F^*[L]$ are those of $F^*[L_i]$ together with those of $X^L$, so passing from $F^*[L]$ to $F^*[L_i]$ changes the degree parity exactly at the vertices of $\operatorname{odd}(X^L)$; this gives the second identity.

Moreover, since $|I_L^*|\le |I|+h$, $|X^L|\le hq$, $q\le\operatorname{width}(\mathcal L)$, and each edge of $X^L$ has at most one endpoint in $L_i$, we have $|I_{L_i}^*|=|(I_L^*\cup V(X^L))\cap L_i|\le |I_L^*|+|X^L|\le |I|+h+hq\le B$. Since $|Q_{L_i}^*|$ is even, the boundary multiset $X^L$ is not rejected. When the algorithm enumerates partitions of the prescribed vertex sets $I_{L_i}^*$, it includes $\mathcal P_{L_i}^*$ for every $i\in\{0,\ldots,q\}$. Consequently, the child decomposition $L_1,\ldots,L_q$, the boundary multiset $X^L$, and the local interfaces $\Psi_{L_i}^*$ on $G[L_i]$ induced by $F^*$ are all considered when $Z[L,\Psi_L^*]$ is computed.

Using these local interfaces, let $\widehat F_0$ be the solution returned by $\mathcal A_\beta$ for the local PC-$\Psi_{L_0}^*$-TSP instance on $G[L_0]$, and put $\widehat F_i:=Z[L_i,\Psi_{L_i}^*]$ for $i\in\{1,\ldots,q\}$. Algorithm~\ref{alg:dp} forms $\widehat F_L:=X^L\uplus \widehat F_0\uplus \widehat F_1\uplus\cdots\uplus \widehat F_q$.

\begin{claim}
\label{clm:dp-candidate-feasible}
The edge multiset $\widehat F_L$ is feasible for the local PC-$\Psi_L^*$-TSP instance on $G[L]$.
\end{claim}

We prove Claim~\ref{clm:dp-candidate-feasible} after this proof. Hence $Z[L,\Psi_L^*]$ has local objective value at most that of $\widehat F_L$.

We now bound the local objective value of $\widehat F_L$. Since $L_1,\ldots,L_q$ are the inclusionwise maximal sets among all $S\in\mathcal L(F^*,h)$ with $S\subsetneq L$, every set $S\in\mathcal L(F^*,h)$ with $S\subsetneq L$ is contained in some $L_i$. Hence $F^*[L_0]\cap F^*(\mathcal L,h)=\emptyset$. Together with $X^L\subseteq F^*[L]\cap F^*(\mathcal L,h)$, the definition of $X^L$ gives $\ell(F^*[L]\cap F^*(\mathcal L,h))=\ell(X^L)+\sum_{i=1}^q\ell(F^*[L_i]\cap F^*(\mathcal L,h))$. Since $F^*[L_0]$ is feasible for the local PC-$\Psi_{L_0}^*$-TSP instance on $G[L_0]$, the approximation guarantee of $\mathcal A_\beta$ gives
\[
    \ell(\widehat F_0)+\pi_{L_0}^{\Psi_{L_0}^*}(\widehat F_0)
    \le
    \beta\bigl(\ell(F^*[L_0])+\pi_{L_0}^{\Psi_{L_0}^*}(F^*[L_0])\bigr)
    =
    \beta\bigl(\ell(F^*[L_0]\setminus F^*(\mathcal L,h))+\pi_{L_0}^{\Psi_{L_0}^*}(F^*[L_0])\bigr).
\]
Applying the induction hypothesis to $L_1,\ldots,L_q$ gives
\[
\sum_{i=1}^q
\bigl(\ell(\widehat F_i)+\pi_{L_i}^{\Psi_{L_i}^*}(\widehat F_i)\bigr)
\le
\sum_{i=1}^q
\left[
\beta\bigl(\ell(F^*[L_i]\setminus F^*(\mathcal L,h))
+\pi_{L_i}^{\Psi_{L_i}^*}(F^*[L_i])\bigr)
+\ell(F^*[L_i]\cap F^*(\mathcal L,h))
\right].
\]
The penalties decompose over the parts $L_0,L_1,\ldots,L_q$: the vertices that are prescribed in the local instance on some $L_i$ but not in the local instance on $L$ are endpoints of edges of $X^L$, and these vertices are visited by both $\widehat F_L$ and $F^*[L]$, so they pay no penalty in any of these instances. Since $L_0,L_1,\ldots,L_q$ partition $L$, we have
\[
\pi_L^{\Psi_L^*}(\widehat F_L)=\sum_{i=0}^q\pi_{L_i}^{\Psi_{L_i}^*}(\widehat F_i)
\qquad\text{and}\qquad
\pi_L^{\Psi_L^*}(F^*[L])=\sum_{i=0}^q\pi_{L_i}^{\Psi_{L_i}^*}(F^*[L_i]).
\]
Also, the edge multisets $F^*[L_0],F^*[L_1],\ldots,F^*[L_q],X^L$ form a partition of $F^*[L]$. Combining these facts, we obtain
\[
\begin{aligned}
    \ell(\widehat F_L)+\pi_L^{\Psi_L^*}(\widehat F_L)
    &=
    \ell(X^L)+
    \sum_{i=0}^q
    \bigl(\ell(\widehat F_i)+\pi_{L_i}^{\Psi_{L_i}^*}(\widehat F_i)\bigr)\\
    &\le
    \sum_{i=0}^q
    \beta\bigl(\ell(F^*[L_i]\setminus F^*(\mathcal L,h))
    +\pi_{L_i}^{\Psi_{L_i}^*}(F^*[L_i])\bigr)
    +\ell(F^*[L]\cap F^*(\mathcal L,h))\\
    &=
    \beta\bigl(\ell(F^*[L]\setminus F^*(\mathcal L,h))
        +\pi_L^{\Psi_L^*}(F^*[L])\bigr)
    +\ell(F^*[L]\cap F^*(\mathcal L,h)).
\end{aligned}
\]
The same bound holds for $Z[L,\Psi_L^*]$. This completes the induction.

Taking $L=V$, and using $\Psi_V^*=\Phi$, $F^*[V]=F^*$, gives
\[
\begin{aligned}
    \ell(Z[V,\Phi])+\pi_V^\Phi(Z[V,\Phi])
    &\le
    \beta\bigl(\ell(F^*)-\ell(F^*(\mathcal L,h))+\pi_V^\Phi(F^*)\bigr)
    +\ell(F^*(\mathcal L,h))\\
    &=
    \beta\,\mathrm{OPT}-(\beta-1)\ell(F^*(\mathcal L,h)).
\end{aligned}\qedhere
\]
\end{proof}

\begin{proof}[Proof of Claim~\ref{clm:dp-candidate-feasible}]
We verify the three feasibility requirements for the local PC-$\Psi_L^*$-TSP instance on $G[L]$.

For the parity constraint, fix $i\in\{0,\ldots,q\}$ and a vertex $v\in L_i$. Since $\widehat F_i$ and $F^*[L_i]$ are both feasible for the local PC-$\Psi_{L_i}^*$-TSP instance on $G[L_i]$, both have $Q_{L_i}^*$ as their set of odd-degree vertices. At $v$, the edges of $\widehat F_L$ are those of $\widehat F_i$ together with those of $X^L$, and the edges of $F^*[L]$ are those of $F^*[L_i]$ together with those of $X^L$. Hence $v$ has the same degree parity in $\widehat F_L$ as in $F^*[L]$. Since $\operatorname{odd}(F^*[L])=Q_L^*$, we obtain $\operatorname{odd}(\widehat F_L)=Q_L^*$.

We next prove the following connectivity preservation statement: if two vertices $u,v\in I_L^*\cup V(X^L)$ lie in the same connected component of $(L,F^*[L])$, then they also lie in the same connected component of $(L,\widehat F_L)$. To prove this, take a $u$-$v$ path $P$ in $(L,F^*[L])$ and split $P$ at the edges of $X^L$. Each maximal remaining subpath is contained in some part $L_i$, and its endpoints belong to $(I_L^*\cup V(X^L))\cap L_i=I_{L_i}^*$. These endpoints are connected in $(L_i,F^*[L_i])$, so by the definition of $\Psi_{L_i}^*$ they belong to the same set of $\mathcal P_{L_i}^*$. Since $\widehat F_i$ is feasible for the local PC-$\Psi_{L_i}^*$-TSP instance on $G[L_i]$, the same endpoints are connected in $(L_i,\widehat F_i)$. Replacing each such subpath of $P$ in $(L_i,F^*[L_i])$ by a corresponding path in $(L_i,\widehat F_i)$, and keeping the edges of $P\cap X^L$, gives a $u$-$v$ walk in $(L,\widehat F_L)$, proving the preservation statement.

The $\mathcal P_L^*$-connectivity constraint follows from the preservation statement: if two vertices of $I_L^*$ belong to the same set of $\mathcal P_L^*$, then they are connected in $(L,F^*[L])$, and hence also in $(L,\widehat F_L)$.

It remains to verify the $I_L^*$-connectivity constraint. Fix an arbitrary vertex $w\in I_L^*\cup V(\widehat F_L)$, and let $D$ be the connected component of $(I_L^*\cup V(\widehat F_L),\widehat F_L)$ containing $w$. Choose $i\in\{0,\ldots,q\}$ with $w\in L_i$. Since $V(\widehat F_L)\cap L_i\subseteq V(X^L)\cup V(\widehat F_i)$ and $(I_L^*\cup V(X^L))\cap L_i=I_{L_i}^*$, we have $w\in I_{L_i}^*\cup V(\widehat F_i)$. Since $\widehat F_i$ is feasible for $\Psi_{L_i}^*$, the connected component of $(I_{L_i}^*\cup V(\widehat F_i),\widehat F_i)$ containing $w$ contains a vertex $w'\in I_{L_i}^*=(I_L^*\cup V(X^L))\cap L_i$. Hence $D$ contains a vertex $w'\in I_L^*\cup V(X^L)$. If $w'\in I_L^*$, then $D$ contains a vertex of $I_L^*$. Otherwise, $w'\in V(X^L)$. Since $F^*[L]$ is feasible for $\Psi_L^*$, the vertex $w'$ and some vertex of $I_L^*$ lie in the same connected component of $(I_L^*\cup V(F^*[L]),F^*[L])$, and hence also in the same connected component of $(L,F^*[L])$. By the preservation statement above, $w'$ lies in a connected component of $(L,\widehat F_L)$ containing a vertex of $I_L^*$. Since $w'\in D$, the component $D$ contains a vertex of $I_L^*$. Since $w$ was arbitrary, every connected component of $(I_L^*\cup V(\widehat F_L),\widehat F_L)$ contains a vertex of $I_L^*$.

Therefore $\widehat F_L$ satisfies the parity, $\mathcal P_L^*$-connectivity, and $I_L^*$-connectivity constraints, and is feasible for the local PC-$\Psi_L^*$-TSP instance on $G[L]$.
\end{proof}

\begin{lemma}
\label{lem:dual-charging}
Let $\mathcal L$, $h$, and $\mathcal A_\beta$ be as in Lemma~\ref{lem:dp-guarantee}, and suppose that $\mathcal L=\{S\in\mathcal S:y_S>0\}$ for a family $\mathcal S\subseteq 2^V$ and a feasible solution $y$ to $(\mathcal S\text{-}\mathrm{D})$. If $F^*$ crosses every set in $\mathcal L$ (the crossing condition), then the output $Z[V,\Phi]$ of Algorithm~\ref{alg:dp} satisfies
\[
    \ell(Z[V,\Phi])+\pi_V^\Phi(Z[V,\Phi])
    \le
    \beta\,\mathrm{OPT}
    -(\beta-1)\left(\sum_{S\in\mathcal L}y_S-\frac{\mathrm{OPT}}{h+1}\right).
\]
\end{lemma}

\begin{proof}
Let $\mathcal L_{\le h}:=\{S\in\mathcal L:1\le |F^*\cap\delta(S)|\le h\}$ and $\mathcal L_{>h}:=\{S\in\mathcal L:|F^*\cap\delta(S)|>h\}$. By the crossing condition, $\mathcal L_{\le h}$ and $\mathcal L_{>h}$ partition $\mathcal L$. Since $y$ is feasible for $(\mathcal S\text{-}\mathrm{D})$, we have
$\sum_{S\in\mathcal S:e\in\delta(S)}y_S\le \ell(e)$ for every edge $e$. Also, every set in $\mathcal L_{>h}$ is crossed by $F^*$ at least $h+1$ times, while every set in $\mathcal L_{\le h}$ is crossed by $F^*(\mathcal L,h)$ at least once. Hence
\[
    (h+1)\sum_{S\in\mathcal L_{>h}}y_S
    \le \sum_{S\in\mathcal L_{>h}} |F^*\cap\delta(S)|\,y_S = \sum_{e\in F^*}\sum_{\substack{S\in\mathcal L_{>h}\\ e\in\delta(S)}}y_S \le \sum_{e\in F^*}\sum_{\substack{S\in\mathcal S\\ e\in\delta(S)}}y_S \le \ell(F^*),
\]
\[
    \ell(F^*(\mathcal L,h))
    \ge \sum_{e\in F^*(\mathcal L,h)}\sum_{\substack{S\in\mathcal S\\ e\in\delta(S)}}y_S \ge \sum_{e\in F^*(\mathcal L,h)}\sum_{\substack{S\in\mathcal L_{\le h}\\ e\in\delta(S)}}y_S = \sum_{S\in\mathcal L_{\le h}} |F^*(\mathcal L,h)\cap\delta(S)|\,y_S \ge \sum_{S\in\mathcal L_{\le h}}y_S.
\]
Therefore
$\ell(F^*(\mathcal L,h))\ge \sum_{S\in\mathcal L}y_S-\ell(F^*)/(h+1)$. Since $\ell(F^*)\le\mathrm{OPT}$ and $\beta\ge1$, Lemma~\ref{lem:dp-guarantee} gives
\[
\begin{aligned}
    \ell(Z[V,\Phi])+\pi_V^\Phi(Z[V,\Phi])
    \le \beta\,\mathrm{OPT}
    -(\beta-1)\left(\sum_{S\in\mathcal L}y_S-\frac{1}{h+1}\mathrm{OPT}\right).
\end{aligned}\qedhere
\]
\end{proof}

For $\mathcal L_{\mathcal D}$, every set $S\in\mathcal L_{\mathcal D}$ separates a demand pair whose endpoints are required by $\mathcal P$ to be connected. Hence every feasible solution to the PC-$\Phi$-TSP instance crosses every set in $\mathcal L_{\mathcal D}$, so the crossing condition holds and Lemma~\ref{lem:dual-charging} applies with $\mathcal S=\mathcal S_{\mathcal D}$ and $y=y^{\mathcal D}$. For $\mathcal L_T$, the crossing condition need not hold. If $T$ were equal to $Q$, every feasible solution $F$ would cross every set $S\in\mathcal L_T$, since $|S\cap Q|$ would be odd and hence, by the handshake lemma, so would $|F\cap\delta(S)|$. However, $T=\operatorname{odd}(H)\triangle Q$, where $H$ is constructed in Algorithm~\ref{alg:basic}, is a subset of $I$ but may differ from $Q$. Thus Lemma~\ref{lem:dual-charging} may not apply to $\mathcal L_T$. The next section handles this case by the recursive improvement algorithm.

\begin{remark}
\label{rem:pcs-crossing}
For a PCS instance, $I=Q=\{s,t\}$, so $T\in\{\emptyset,Q\}$, and every feasible solution crosses every set in $\mathcal L_T$ and in $\mathcal L_{\mathcal D}$. The local instances created by the dynamic program, however, need not have this property. Hence, even for PCS, the iteration relies on the recursive improvement algorithm of Section~\ref{sec:recursive}.
\end{remark}

\section{The recursive improvement algorithm}
\label{sec:recursive}

We now describe the recursive improvement algorithm. Given a PC-$\Phi$-TSP instance, an input approximation algorithm $\mathcal A_\beta$ for PC-$\Phi$-TSP, and an error parameter $\eta$, set $h=\lfloor10/\eta\rfloor$ for the calls to Algorithm~\ref{alg:dp}. The algorithm first computes three nonrecursive candidates: the output of Algorithm~\ref{alg:basic} and the two outputs of Algorithm~\ref{alg:dp} on the laminar families $\mathcal L_{\mathcal D}$ and $\mathcal L_T$, respectively, with $\mathcal A_\beta$ used as a subroutine. It then tries to split the instance into two smaller instances. For each such split, it solves the two smaller instances recursively and adds the multiunion of the returned edge multisets as another candidate. The algorithm returns the best candidate found. The complete procedure is summarized in Algorithm~\ref{alg:recursive}.

It remains to specify the splits considered by the algorithm. Consider the current instance in a recursive call, with graph $G=(V,E)$ and interface $\Phi=(I,Q,\mathcal P)$. For $R\subseteq V$, let $\bar R:=V\setminus R$. We call $R$ eligible for $\Phi$ if $1\le |I\cap R|\le |I|-1$, $|Q\cap R|$ is even, and each $C\in\mathcal P$ is contained in either $R$ or $\bar R$. Figure~\ref{fig:eligible_split} illustrates some eligible sets. For an eligible $R$, define $\Phi^R:=\bigl(I\cap R,\ Q\cap R,\ \{C\in\mathcal P:C\subseteq R\}\bigr)$ and $\Phi^{\bar R}:=\bigl(I\cap\bar R,\ Q\cap\bar R,\ \{C\in\mathcal P:C\subseteq\bar R\}\bigr)$. We call $\Phi^R$ and $\Phi^{\bar R}$ the restricted interfaces induced by $R$. The corresponding restricted instances are defined on $G[R]$ and $G[\bar R]$ with interfaces $\Phi^R$ and $\Phi^{\bar R}$, respectively, and with edge lengths and vertex penalties inherited from the current instance. We call this split the eligible split induced by $R$. By eligibility, $1\le |I\cap R|\le |I|-1$ and $1\le |I\cap\bar R|\le |I|-1$, so both restricted instances have fewer prescribed vertices than the current instance. The other two conditions ensure that $\Phi^R$ and $\Phi^{\bar R}$ are interfaces: $|Q\cap R|$ and $|Q\cap\bar R|=|Q|-|Q\cap R|$ are even, and the sets of $\mathcal P$ contained in $R$ and in $\bar R$ form partitions of $I\cap R$ and $I\cap\bar R$, respectively. Moreover, for edge multisets $F^R\subseteq E(R)$ and $F^{\bar R}\subseteq E(\bar R)$, the multiunion $F^R\uplus F^{\bar R}$ is feasible for the current instance if and only if $F^R$ and $F^{\bar R}$ are feasible for the two restricted instances, and its objective value is the sum of their objective values.

\begin{figure}[htpb]\small 
\centering
\includegraphics[height=0.4\textwidth,keepaspectratio]{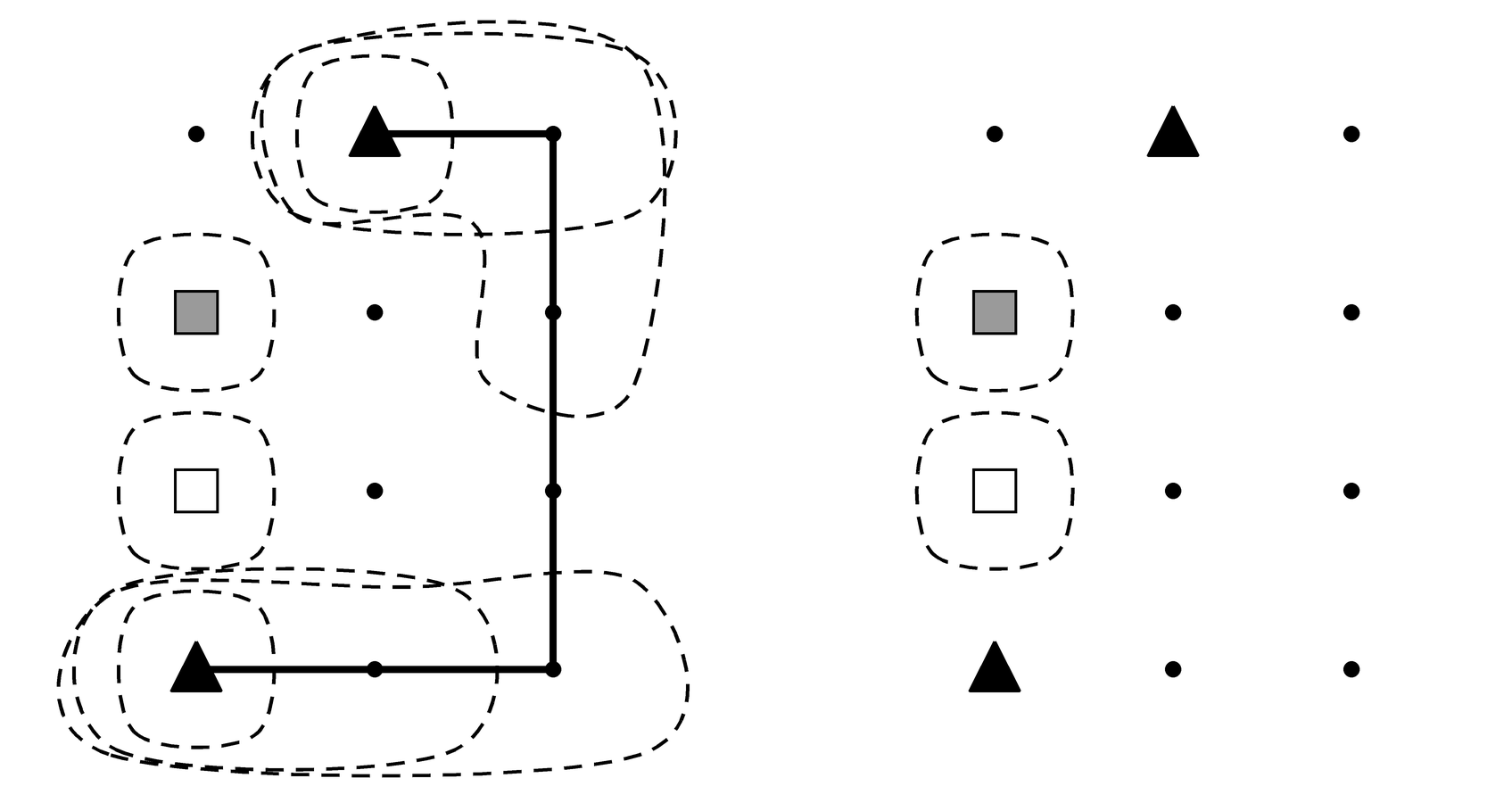}
\caption{An illustration of eligible sets. The left panel shows an optimum solution and the sets of a laminar family. The right panel shows the eligible sets in this laminar family.}
\label{fig:eligible_split}
\end{figure}

\begin{lemma}
\label{lem:eligible-decomposition}
Let $\mathcal L_T$ be the laminar family computed by Algorithm~\ref{alg:basic} on a PC-$\Phi$-TSP instance with $\Phi=(I,Q,\mathcal P)$ and optimum value $\mathrm{OPT}$. If some optimum solution does not cross a set $U\in\mathcal L_T$, then $U$ is eligible for $\Phi$, and the optimum values of the two restricted instances induced by $U$ sum to at most $\mathrm{OPT}$.
\end{lemma}

\begin{proof}
Let $F^*$ be an optimum solution with $F^*\cap\delta(U)=\emptyset$. Recall from Section~\ref{sec:basic} that $T\subseteq I$ has even cardinality. Since $U\in\mathcal L_T\subseteq\mathcal S_T$, the number $|T\cap U|$ is odd, so both $U$ and $\bar U$ contain vertices of $T$, and hence of $I$. By the parity constraint, $\operatorname{odd}(F^*)=Q$, so the handshake lemma applied to $U$ gives $|Q\cap U|\equiv|F^*\cap\delta(U)|=0\pmod 2$. No set $C\in\mathcal P$ has vertices in both $U$ and $\bar U$, since $F^*$ connects all vertices of $C$ and has no edge in $\delta(U)$. Thus $U$ is eligible. Finally, $F^*=F^*[U]\uplus F^*[\bar U]$, so $F^*[U]$ and $F^*[\bar U]$ are feasible for the two restricted instances and their objective values sum to $\mathrm{OPT}$.
\end{proof}

\begin{algorithm}[!ht]
\caption{\textsc{RecursiveImprove}$(G,\ell,\pi,\Phi,\eta,\mathcal A_\beta)$}
\label{alg:recursive}
\LinesNumbered
\noindent\textbf{Input:} a PC-$\Phi$-TSP instance with $\Phi=(I,Q,\mathcal P)$, a $\beta$-approximation algorithm $\mathcal A_\beta$ for PC-$\Phi$-TSP that runs in polynomial time on instances whose number of prescribed vertices is bounded by a fixed constant, and an error parameter $0<\eta\le1$.\\
\noindent\textbf{Output:} a feasible solution to the input instance.
\BlankLine
\nl $h\gets \lfloor10/\eta\rfloor$\;
\nl $(F_{\mathrm{bas}},\mathcal L_T,\mathcal L_{\mathcal D})\gets
\textsc{BasicApprox}(G,\ell,\pi,\Phi)$\;
\nl $F_{\mathcal D}\gets\textsc{LaminarDP}(G,\ell,\pi,\Phi,\mathcal L_{\mathcal D},h,\mathcal A_\beta)$\;
\nl $F_T\gets\textsc{LaminarDP}(G,\ell,\pi,\Phi,\mathcal L_T,h,\mathcal A_\beta)$\;
\nl $\mathcal F\gets\{F_{\mathrm{bas}},F_{\mathcal D},F_T\}$\;
\nl \For{each $R\in\mathcal L_T$ that is eligible for $\Phi$}{
    Form the restricted interfaces $\Phi^R$ and $\Phi^{\bar R}$\;
    $F^R\gets\textsc{RecursiveImprove}(G[R],\ell,\pi,\Phi^R,\eta,\mathcal A_\beta)$\;
    $F^{\bar R}\gets\textsc{RecursiveImprove}(G[\bar R],\ell,\pi,\Phi^{\bar R},\eta,\mathcal A_\beta)$\;
    $\mathcal F\gets \mathcal F\cup\{F^R\uplus F^{\bar R}\}$\;
  }
\nl \Return a member $F\in\mathcal F$ minimizing $\ell(F)+\pi_V^\Phi(F)$\;
\end{algorithm} 

Let $n$ be the number of vertices in the input instance. On instances with $|I|\le k$, the running time of Algorithm~\ref{alg:recursive} is bounded by the sum of the running times of the PCTSP algorithm and of $\mathcal A_\beta$ on instances with at most $k+h(k+1)$ prescribed vertices, multiplied by a polynomial in $n$ whose exponent depends only on $k$ and $\eta$. Indeed, by laminarity, the family $\mathcal L_T$ computed in any current instance contains $O(n)$ sets, so each call considers $O(n)$ eligible sets and creates at most $O(n)$ further recursive calls. By eligibility, every recursive call strictly decreases the number of prescribed vertices, so the recursion depth is at most $k$, and the total number of recursive calls is at most $n^{O(k)}$. In each recursive call, Algorithm~\ref{alg:basic} calls the PCTSP algorithm at most once and otherwise runs in polynomial time, and, since the current instance has at most $k$ prescribed vertices, the two calls to Algorithm~\ref{alg:dp} satisfy the running-time bound for Algorithm~\ref{alg:dp} with $h=\lfloor10/\eta\rfloor$. All exponents in these bounds depend only on $k$ and $\eta$.

\begin{lemma}
\label{lem:one-step}
Let $0<\eta\le1$, and let $\mathcal A_\beta$ be a $\beta$-approximation algorithm for PC-$\Phi$-TSP. Then, on any PC-$\Phi$-TSP instance with optimum value $\mathrm{OPT}$, Algorithm~\ref{alg:recursive} with input $\eta$ and $\mathcal A_\beta$ returns a solution of objective value at most $\beta'\mathrm{OPT}$, where $\beta':=\max\{\rho+\eta,\beta-\frac{\eta}{10}(\beta-1)\}$ and $\rho$ is the approximation ratio of the PCTSP algorithm used in Algorithm~\ref{alg:basic}.
\end{lemma}

\begin{proof}
Fix an input instance and consider the execution of Algorithm~\ref{alg:recursive} on this input. We prove a stronger statement: for every recursive call in this execution, the solution returned by the call has objective value at most $\beta'$ times the optimum value of its current instance. We prove this statement by induction, in increasing order of the number of prescribed vertices in the current instance of the recursive call.

Consider one recursive call of Algorithm~\ref{alg:recursive} on input $(G,\ell,\pi,\Phi,\eta,\mathcal A_\beta)$, where the current instance has interface $\Phi=(I,Q,\mathcal P)$. Let $\mathrm{OPT}_0$ be the optimum value of the current instance. If $I=\emptyset$, then the candidate produced by Algorithm~\ref{alg:basic} is the empty edge multiset, which is optimal, and the desired bound is immediate. This proves the base case of the induction. We may therefore assume that $I\ne\emptyset$ and that the statement holds for all recursive calls in this execution whose current instances have fewer than $|I|$ prescribed vertices.

Run Algorithm~\ref{alg:basic} on the current instance, and use its notation
$F_{\mathrm{bas}}$, $\mathcal L_T$, $\mathcal L_{\mathcal D}$,
$H$, $J$, $K$, $T$, $y^T$, and $y^{\mathcal D}$. Let $F_{\mathcal D}$ and $F_T$ be the two candidates obtained by applying Algorithm~\ref{alg:dp} to the current instance with laminar families $\mathcal L_{\mathcal D}$ and $\mathcal L_T$, respectively, using $\mathcal A_\beta$ as a subroutine.

Consider first the case in which the crossing condition holds for $\mathcal L_T$, that is, the current instance has an optimum solution $F_0^*$ crossing every set of $\mathcal L_T$. Since every feasible solution crosses every set of $\mathcal L_{\mathcal D}$, the solution $F_0^*$ crosses all sets in both laminar families. Recall that Algorithm~\ref{alg:recursive} uses $h=\lfloor10/\eta\rfloor$, and hence $1/(h+1)<\eta/10$. If $\ell(J)\le(\eta/5)\mathrm{OPT}_0$ and $\sum_{S\in\mathcal L_{\mathcal D}}y^{\mathcal D}_S\le(\eta/5)\mathrm{OPT}_0$, then Lemma~\ref{lem:basic-guarantee} gives
\[
    \ell(F_{\mathrm{bas}})+\pi_V^\Phi(F_{\mathrm{bas}})
    \le \rho\mathrm{OPT}_0+\ell(J)+4\sum_{S\in\mathcal L_{\mathcal D}}y^{\mathcal D}_S
    \le (\rho+\eta)\mathrm{OPT}_0
    \le \beta'\mathrm{OPT}_0.
\]
If $\sum_{S\in\mathcal L_{\mathcal D}}y^{\mathcal D}_S>(\eta/5)\mathrm{OPT}_0$, then Lemma~\ref{lem:dual-charging}, applied to $F_{\mathcal D}$ with $\mathcal S=\mathcal S_{\mathcal D}$, $y=y^{\mathcal D}$, and $F^*=F_0^*$, gives
\[
    \ell(F_{\mathcal D})+\pi_V^\Phi(F_{\mathcal D})
    \le \beta\mathrm{OPT}_0-(\beta-1)\left(\sum_{S\in\mathcal L_{\mathcal D}}y^{\mathcal D}_S-\frac{\mathrm{OPT}_0}{h+1}\right)
    \le \left(\beta-\frac{\eta}{10}(\beta-1)\right)\mathrm{OPT}_0
    \le \beta'\mathrm{OPT}_0.
\]
If $\sum_{S\in\mathcal L_{\mathcal D}}y^{\mathcal D}_S\le(\eta/5)\mathrm{OPT}_0$ and $\ell(J)=\sum_{S\in\mathcal L_T} y^T_S>(\eta/5)\mathrm{OPT}_0$, then Lemma~\ref{lem:dual-charging}, applied to $F_T$ with $\mathcal S=\mathcal S_T$, $y=y^T$, and $F^*=F_0^*$, gives
\[
    \ell(F_T)+\pi_V^\Phi(F_T)
    \le \beta\mathrm{OPT}_0-(\beta-1)\left(\ell(J)-\frac{\mathrm{OPT}_0}{h+1}\right)
    \le \left(\beta-\frac{\eta}{10}(\beta-1)\right)\mathrm{OPT}_0
    \le \beta'\mathrm{OPT}_0.
\]
Thus, whenever the crossing condition holds for $\mathcal L_T$, one of the three nonrecursive candidates has objective value at most $\beta'\mathrm{OPT}_0$.

We are left with the case in which no optimum solution of the current instance crosses every set of $\mathcal L_T$. Fix an optimum solution $F_0^*$ and a set $U\in\mathcal L_T$ with $F_0^*\cap\delta(U)=\emptyset$. By Lemma~\ref{lem:eligible-decomposition}, $U$ is eligible, so Algorithm~\ref{alg:recursive} considers the eligible split induced by $U$, and the optimum values $\mathrm{OPT}^U$ and $\mathrm{OPT}^{\bar U}$ of the two restricted instances satisfy $\mathrm{OPT}^U+\mathrm{OPT}^{\bar U}\le\mathrm{OPT}_0$. Both restricted instances have fewer prescribed vertices than the current instance, so the induction hypothesis applies to the two recursive calls. Let $F^U$ and $F^{\bar U}$ be the returned solutions. Then $F^U\uplus F^{\bar U}$ is feasible for the current instance, is a candidate considered by Algorithm~\ref{alg:recursive}, and has objective value at most $\beta'\mathrm{OPT}^U+\beta'\mathrm{OPT}^{\bar U}\le\beta'\mathrm{OPT}_0$.

Therefore, in all cases, Algorithm~\ref{alg:recursive} considers a candidate for the current instance whose objective value is at most $\beta'\mathrm{OPT}_0$. Since the algorithm returns the best candidate, its output on the current instance also has objective value at most $\beta'\mathrm{OPT}_0$. This completes the induction. Applying the statement to the initial recursive call proves the lemma.
\end{proof}

We now prove Theorem~\ref{thm:main}. Following the iterative scheme of \cite{TVZ2022}, we define a sequence of algorithms: the algorithm at one level is used as the subroutine $\mathcal A_\beta$ in Algorithm~\ref{alg:recursive} to obtain the algorithm at the next level.

\begin{proof}[Proof of Theorem~\ref{thm:main}]
Fix $\varepsilon>0$. Set $\eta:=\min\{1,\varepsilon\}$ and $\beta_0:=2\rho+5$. Choose
\[
    N:=\left\lceil
    \frac{\log\bigl((\rho+\eta-1)/(\beta_0-1)\bigr)}
         {\log(1-\eta/10)}
    \right\rceil .
\]
Then $N$ is a fixed constant, and $1+(1-\eta/10)^N(\beta_0-1)\le \rho+\eta$. Let $\mathcal A_0$ be the algorithm that runs Algorithm~\ref{alg:basic} and outputs only $F_{\mathrm{bas}}$. For $j\in\{1,\ldots,N\}$, define $\mathcal A_j$ to be the algorithm that runs Algorithm~\ref{alg:recursive} on input $(G,\ell,\pi,\Phi,\eta,\mathcal A_{j-1})$. By Lemma~\ref{lem:basic-guarantee}, $\mathcal A_0$ is a $\beta_0$-approximation algorithm. Repeatedly applying Lemma~\ref{lem:one-step}, the approximation ratio of $\mathcal A_j$ for $j\in\{1,\ldots,N\}$ is at most $\beta_j$, where
\[
    \beta_{j}:=\max\left\{\rho+\eta,\ \beta_{j-1}-\frac{\eta}{10}(\beta_{j-1}-1)\right\}.
\]
Since $\beta-\frac{\eta}{10}(\beta-1)=1+(1-\frac{\eta}{10})(\beta-1)$ for every $\beta$, induction on $j$ gives, for $j\in\{1,\ldots,N\}$,
\[
    \beta_j\le \max\left\{\rho+\eta,\ 1+(1-\eta/10)^j(\beta_0-1)\right\}.
\]
Hence, by the choice of $N$, the algorithm $\mathcal A_N$ has approximation ratio at most $\rho+\eta\le\rho+\varepsilon$.

It remains to verify the running time. By Lemma~\ref{lem:basic-guarantee}, the running time of $\mathcal A_0$ is bounded by the running time of the PCTSP algorithm multiplied by a polynomial in $n$. By the running-time bound for Algorithm~\ref{alg:recursive}, if $\mathcal A_{j-1}$ runs in polynomial time on instances whose number of prescribed vertices is bounded by a constant, then so does $\mathcal A_j$; since $N$ depends only on $\rho$ and $\varepsilon$, so does $\mathcal A_N$. More precisely, applying the running-time bound for Algorithm~\ref{alg:recursive} to $\mathcal A_N,\mathcal A_{N-1},\ldots,\mathcal A_1$ in turn, starting with at most $k$ prescribed vertices, shows that $\mathcal A_N$ runs in time $n^{f(k,1/\varepsilon)}$ times the running time of the PCTSP algorithm, for a function $f$ depending on $\rho$. Thus $\mathcal A_N$ is the algorithm claimed in Theorem~\ref{thm:main}. This proves the theorem.
\end{proof}

\section{Conclusion}
\label{sec:conclusion}

We have shown that PC-$\Phi$-TSP instances whose number of prescribed vertices is bounded by a fixed constant can be approximated as well as PCTSP, up to an arbitrarily small additive loss in the approximation ratio. Hence every approximation guarantee for PCTSP transfers, up to such a loss, to PCS, to the other prize-collecting routing problems captured by such instances, and, by Corollary~\ref{cor:multipath}, to the prize-collecting multi-path TSP with a fixed number of terminal pairs.

The bound on the number of prescribed vertices keeps the enumeration in the dynamic programs and the number of recursive calls of Algorithm~\ref{alg:recursive} polynomially bounded. We do not know whether the same transfer holds when the number of prescribed vertices, or the number of terminal pairs, is unbounded.

\paragraph{Acknowledgment.}
OpenAI's ChatGPT and Anthropic's Claude were used to assist with the writing of this paper and to check the correctness and consistency of its arguments. The author takes full responsibility for the content of the paper.

\begingroup
\setlength{\bibsep}{6pt plus 1pt minus 1pt}

\endgroup

\clearpage
\appendix
\renewcommand{\theHsection}{appendix.\Alph{section}}

\section{Extension to the prize-collecting multi-path TSP}
\label{app:multipath}
 
We prove Corollary~\ref{cor:multipath} by adapting the reduction of the multi-path TSP to TSP by Traub and Vygen~\cite[Sections~16.4--16.7]{TraubVygen2024}, within the iterative scheme of Sections~\ref{sec:basic}--\ref{sec:recursive}. Their reduction combines two algorithms, whose performance depends on the total distance, the sum of the distances $\ell(s_it_i)$ over all terminal pairs $\{s_i,t_i\}$. The first reduces the instance to TSP~\cite[Lemma~16.13]{TraubVygen2024} and is good when the total distance is small compared with the optimum value. The second is a dynamic program that takes an approximation algorithm for the multi-path TSP as input and is good when the total distance is large. It guesses the edges of an optimum solution in cuts that this solution crosses only a few times, pays for these edges at their actual length, and applies the input algorithm to the remaining parts, which are again instances of the multi-path TSP. The cuts come from one chain of cuts for each terminal pair $\{s_i,t_i\}$ with $s_i\ne t_i$, which carries a feasible dual solution, in the sense of Section~\ref{sec:prelim}, of value $\ell(s_it_i)$. The analysis uses these dual solutions one terminal pair at a time, and their values sum to the total distance. Since the chains of different pairs may cross, their sets need not form a laminar family, and the dynamic program works with vertex sets of a special form, called boxed sets.
 
The first algorithm relies on the requirement that every vertex is visited, so our first candidate replaces it by an algorithm based on Theorem~\ref{thm:main}. A solution to the prize-collecting multi-path TSP consists of one walk for each terminal pair. Taken together, the edges of these walks form an edge multiset that satisfies the parity and connectivity requirements of a PC-$\Phi$-TSP instance whose prescribed vertices are the terminals. The converse fails, since realizing the prescribed pairing may require more copies of some edges than such a multiset contains. For example, the edges on the right of Figure~\ref{fig:pc_phi_example} have odd degree exactly at the six triangles and connect them; if each triangle at the top is paired with a triangle at the bottom, all three walks would have to use the single path between the two groups. The first candidate solves this PC-$\Phi$-TSP instance using Theorem~\ref{thm:main} and then repairs the pairing by adding two copies of the edge $s_it_i$ for each terminal pair $\{s_i,t_i\}$ with $s_i\ne t_i$, one to make all degrees even and one to form a separate walk between $s_i$ and $t_i$, to which the remaining edges are attached. Its repair cost is thus twice the total distance. As for the basic algorithm, this candidate is good when the repair cost is small compared with the optimum value, and it alone serves as the initial algorithm of the iteration.
 
Our second candidate adapts the dynamic program of Traub and Vygen to the prize-collecting setting. It takes the current approximation algorithm for the prize-collecting multi-path TSP as input, and the objectives of its local instances include the penalties of unvisited vertices, which decompose over the local instances because restricting a solution to a vertex subset preserves exactly the vertices that it visits there. Like the dynamic program of Section~\ref{sec:dp}, it is good when the repair cost is large. In Section~\ref{sec:dp}, however, the cuts come from laminar families whose dual solutions account for the two repair costs, and an optimum solution need not cross all of them. Here, the walk of an optimum solution for a pair crosses all cuts of its chain, so the analogue of the crossing condition of Lemma~\ref{lem:dual-charging} holds automatically, and no recursive splitting as in Section~\ref{sec:recursive} is needed. Each iteration returns the better of the two candidates, which improves the current approximation ratio, and a constant number of iterations proves Corollary~\ref{cor:multipath}. The remaining arguments and our terminology follow~\cite{TraubVygen2024}.
 
\subsection{The model}
 
An instance of the prize-collecting multi-path TSP consists of a complete metric graph $G=(V,E)$ with edge lengths $\ell$, vertex penalties $\pi$, and terminal pairs $\mathcal T=(\{s_i,t_i\})_{i=1}^{\kappa}$. As in~\cite[Problem~16.11]{TraubVygen2024}, each terminal pair $\{s_i,t_i\}$ is a multiset of two vertices, so $s_i=t_i$ is allowed; this generalizes the version described in Section~\ref{sec:intro}. A solution consists of edge multisets $F_1,\ldots,F_\kappa$ such that, for every $i$, $F_i$ is an $s_i$-$t_i$ tour if $s_i\ne t_i$, and $F_i$ is the empty tour or a tour with $s_i\in V(F_i)$ if $s_i=t_i$; equivalently, $\operatorname{odd}(F_i)=\{s_i\}\triangle\{t_i\}$ and $(V(F_i)\cup\{s_i,t_i\},F_i)$ is connected. By Euler's theorem, each $F_i$ can be traversed by a walk $W_i$ from $s_i$ to $t_i$, which consists of the single vertex $s_i=t_i$ if $F_i$ is empty. Since the terminal pairs that a solution induces on a vertex subset, defined below, depend on the order in which these walks traverse their edges, we represent a solution by such a family $\mathcal W=(W_1,\ldots,W_\kappa)$ of walks, as in~\cite[Section~16.6]{TraubVygen2024}. Let $V(\mathcal W)$ be the set of vertices visited by these walks, let $\ell(\mathcal W)$ be their total length, where copies of an edge in different walks are counted separately, and, for $U\subseteq V$, let $\pi_U(\mathcal W):=\pi(U\setminus V(\mathcal W))$. The objective is to minimize $\ell(\mathcal W)+\pi_V(\mathcal W)$, which depends only on the edge multisets $F_1,\ldots,F_\kappa$, and $\mathrm{OPT}$ denotes the optimum value of the multi-path instance under consideration. We also allow instances without terminal pairs; the only solution of such an instance is the empty family, whose objective value is $\pi(V)$.
 
For $U\subseteq V$, the restriction $\mathcal W[U]$ of $\mathcal W$ to $U$ is the family of the maximal subwalks of the walks in $\mathcal W$ whose vertices all lie in $U$, where a single vertex is also a subwalk. Thus $V(\mathcal W[U])=V(\mathcal W)\cap U$, and the end vertices of these subwalks form the terminal pairs induced by $\mathcal W$ on $U$~\cite[Definition~16.14]{TraubVygen2024}.
 
We write $\ell(\mathcal T):=\sum_{i=1}^\kappa\ell(s_it_i)$ for the total distance, where $\ell(vv):=0$. By the triangle inequality, $\ell(\mathcal T)\le\mathrm{OPT}$. After renumbering, we assume that $s_i\ne t_i$ exactly for $i\in\{1,\ldots,\bar\kappa\}$.
 
\subsection{The first candidate: small total distance}
 
Let $I:=\{s_i,t_i:i\in\{1,\ldots,\kappa\}\}$, let $Q$ be the set of vertices that belong to an odd number of the sets $\{s_i\}\triangle\{t_i\}$, and let $\Phi_0:=(I,Q,\{\{v\}:v\in I\})$. Consider the PC-$\Phi_0$-TSP instance on $G$ with the same edge lengths and penalties. The multiunion of the edge multisets of a solution is feasible for this instance and has the same objective value, so the optimum value of this PC-$\Phi_0$-TSP instance is at most $\mathrm{OPT}$.
 
For $0<\eta\le1$, the first candidate, which we call the short-distance candidate, is computed as follows. If $\kappa=0$, it is the empty family, which is the only solution. Otherwise, we run the algorithm of Theorem~\ref{thm:main} with $k=2\kappa$ and error parameter $\eta/2$ on the PC-$\Phi_0$-TSP instance, which returns a feasible solution $F$, and then repair the pairing in two steps. First, let $F^+$ be obtained from $F$ by adding one copy of the edge $s_it_i$ for each $i$ with $s_i\ne t_i$. Then $\operatorname{odd}(F^+)=\emptyset$, every connected component of $(I\cup V(F),F^+)$ contains a vertex of $I$, and $s_i$ and $t_i$ lie in the same component for every $i$; hence every component contains $s_i$ for some $i$. Assign each component to one index $i$ such that it contains $s_i$; if a component contains $s_j$ for several indices $j$, it is assigned to only one of them. For each $i$, let $F^+_i$ consist of the edges of the component assigned to $i$, and let $F^+_i:=\emptyset$ if there is no such component. Thus $F^+=F^+_1\uplus\cdots\uplus F^+_\kappa$, each $F^+_i$ has only even degrees, and $(V(F^+_i)\cup\{s_i\},F^+_i)$ is connected. Second, let $F_i$ be obtained from $F^+_i$ by adding one more copy of the edge $s_it_i$ if $s_i\ne t_i$. Then $\operatorname{odd}(F_i)=\{s_i\}\triangle\{t_i\}$ and $(V(F_i)\cup\{s_i,t_i\},F_i)$ is connected, so, by Euler's theorem, $F_i$ is traversed by a walk $W_i$ from $s_i$ to $t_i$. The short-distance candidate is the solution $\mathcal W:=(W_1,\ldots,W_\kappa)$. It replaces the reduction to TSP in~\cite[Lemma~16.13]{TraubVygen2024}.
 
\begin{lemma}
\label{lem:mp-short}
The short-distance candidate $\mathcal W$ satisfies
\[
\ell(\mathcal W)+\pi_V(\mathcal W)\le(\rho+\eta/2)\mathrm{OPT}+2\ell(\mathcal T).
\]
\end{lemma}
 
\begin{proof}
If $\kappa=0$, the candidate is optimal. Otherwise, $\ell(F)+\pi_V^{\Phi_0}(F)\le(\rho+\eta/2)\mathrm{OPT}$, since the optimum value of the PC-$\Phi_0$-TSP instance is at most $\mathrm{OPT}$. The repair adds edges of total length $2\ell(\mathcal T)$, and the resulting walks visit all terminals and all vertices of $V(F)$, so the penalty does not increase.
\end{proof}
 
Thus the objective value of the short-distance candidate is at most $(\rho+\eta/2)\mathrm{OPT}$ plus the repair cost, as in Lemma~\ref{lem:basic-guarantee} for the basic algorithm.
 
\subsection{The second candidate: large total distance}
 
In this subsection, we assume $\bar\kappa\ge1$. Let $n:=|V|$. For $i\in\{1,\ldots,\bar\kappa\}$, order the vertices as $v^i_1=s_i,v^i_2,\ldots,v^i_n=t_i$ by nondecreasing value of $\min\{\ell(s_iv),\ell(s_it_i)\}$, breaking ties so that $s_i$ comes first and $t_i$ comes last, and let $U^i_j:=\{v^i_1,\ldots,v^i_j\}$ for $j\in\{0,\ldots,n\}$. For $j\in\{1,\ldots,n-1\}$, the cut $\delta(U^i_j)$ separates $s_i$ from $t_i$.
 
For the analysis, fix an optimum solution $\mathcal W^*$, let $F^*$ be the multiunion of the edge multisets of its walks, and fix a positive integer $h$, which plays the same role as in Algorithm~\ref{alg:dp}. We call a cut $\delta(U^i_j)$ with $i\in\{1,\ldots,\bar\kappa\}$ and $j\in\{1,\ldots,n-1\}$ small if $|F^*\cap\delta(U^i_j)|\le h$. Let $E_{\mathrm{small}}$ be the set of edges that lie in at least one small cut, and let $F^*_{\mathrm{small}}:=F^*\cap E_{\mathrm{small}}$, that is, $F^*_{\mathrm{small}}$ consists of all copies of edges of $F^*$ that lie in a small cut; this edge multiset is denoted by $L^*$ in~\cite[(16.6)]{TraubVygen2024}. The dynamic program below aims to guess the edges of $F^*_{\mathrm{small}}$. The next lemma, which plays the role of Lemma~\ref{lem:dual-charging}, shows that they have large total length when the total distance is large.
 
\begin{lemma}
\label{lem:mp-distance}
$\ell(F^*_{\mathrm{small}})\ge\ell(\mathcal T)-\frac{\kappa}{h+1}\,\ell(F^*)$.
\end{lemma}
 
\begin{proof}
We follow the proof of~\cite[Lemma~16.18]{TraubVygen2024}, which is similar to that of Lemma~\ref{lem:dual-charging}. Fix $i\in\{1,\ldots,\bar\kappa\}$, let $f_i(v):=\min\{\ell(s_iv),\ell(s_it_i)\}$ for $v\in V$, and, for $j\in\{1,\ldots,n-1\}$, assign to the set $U^i_j$ the value $y^i_j:=f_i(v^i_{j+1})-f_i(v^i_j)\ge0$. Then $\sum_{j=1}^{n-1}y^i_j=f_i(t_i)-f_i(s_i)=\ell(s_it_i)$. Consider an edge $uv$, and let $u=v^i_{j_1}$ and $v=v^i_{j_2}$ with $j_1<j_2$. Since $U^i_j$ consists of the first $j$ vertices in the order, the cut $\delta(U^i_j)$ contains $uv$ if and only if $j_1\le j<j_2$. Hence
\[
\sum_{j:\,uv\in\delta(U^i_j)}y^i_j=\sum_{j=j_1}^{j_2-1}y^i_j=f_i(v)-f_i(u)\le\ell(uv)
\]
by the triangle inequality; that is, the values $y^i_j$ satisfy the constraints of $(\mathcal S\text{-}\mathrm{D})$ for the family of the sets $U^i_1,\ldots,U^i_{n-1}$. Summing this over all copies of edges in an edge multiset $F$ and exchanging the order of summation, we obtain
\[
\sum_{j=1}^{n-1}|F\cap\delta(U^i_j)|\,y^i_j=\sum_{e\in F}\ \sum_{j:\,e\in\delta(U^i_j)}y^i_j\le\ell(F),
\]
where $e$ runs over all copies of edges in $F$.
 
We apply this inequality twice. First, let $F:=F^*$. Every cut $\delta(U^i_j)$ that is not small contains more than $h$ copies of edges of $F^*$, so
\[
(h+1)\sum_{j:\,\delta(U^i_j)\text{ not small}}y^i_j\le\sum_{j=1}^{n-1}|F^*\cap\delta(U^i_j)|\,y^i_j\le\ell(F^*).
\]
Second, let $F^*_i$ be the edge multiset of the walk $W^*_i$, and let $F:=F^*_i\cap E_{\mathrm{small}}$. Since $W^*_i$ goes from $s_i$ to $t_i$, it uses an edge of every cut $\delta(U^i_j)$, and this edge lies in $E_{\mathrm{small}}$ if the cut is small. Hence $F$ contains an edge of every small cut $\delta(U^i_j)$. Applying the inequality to this $F$, and using $\sum_{j=1}^{n-1}y^i_j=\ell(s_it_i)$ and the first application, we obtain
\[
\ell(F^*_i\cap E_{\mathrm{small}})\ge\sum_{j:\,\delta(U^i_j)\text{ small}}y^i_j=\ell(s_it_i)-\sum_{j:\,\delta(U^i_j)\text{ not small}}y^i_j\ge\ell(s_it_i)-\frac{\ell(F^*)}{h+1}.
\]
 
Finally, $F^*=F^*_1\uplus\cdots\uplus F^*_\kappa$, where $F^*_i$ is the edge multiset of $W^*_i$ for every $i\in\{1,\ldots,\kappa\}$, so $F^*_{\mathrm{small}}=(F^*_1\cap E_{\mathrm{small}})\uplus\cdots\uplus(F^*_\kappa\cap E_{\mathrm{small}})$. Summing the last inequality over $i\in\{1,\ldots,\bar\kappa\}$, and using $\ell(s_it_i)=0$ for $i>\bar\kappa$ and $\bar\kappa\le\kappa$, we obtain
\[
\ell(F^*_{\mathrm{small}})\ge\sum_{i=1}^{\bar\kappa}\ell(F^*_i\cap E_{\mathrm{small}})\ge\ell(\mathcal T)-\frac{\kappa}{h+1}\,\ell(F^*).\qedhere
\]
\end{proof}
 
To guess the edges of $F^*_{\mathrm{small}}$, we adapt the dynamic program of~\cite[Section~16.7]{TraubVygen2024} to the prize-collecting setting. We first describe how the small cuts divide $V$; the small cuts are not known to the algorithm, and the resulting parts are used only in the analysis. For each $i\in\{1,\ldots,\bar\kappa\}$, we call the elements of $\{0,n\}\cup\{j:\delta(U^i_j)\text{ is small}\}$ the slab boundaries for $i$, and the sets $U^i_{j'}\setminus U^i_j$ for consecutive slab boundaries $j<j'$ the slabs for $i$; they partition $V$. We call the nonempty intersections of one slab for each $i$ cells; the slab of a cell for $i$ is the slab for $i$ that contains it. The cells partition $V$, and an edge joins two different cells if and only if it lies in a small cut.
 
We add the cells one at a time to an initially empty set, in the following order. We number the slabs for each $i$ by $1,2,\ldots$ in increasing order of $j$, describe each cell by the vector $(\sigma_1,\ldots,\sigma_{\bar\kappa})$ of the numbers of its slabs, and add the cells in lexicographic order of these vectors. For example, if $\bar\kappa=3$ and there are two slabs for each $i$, the cells are added in the order $(1,1,1),(1,1,2),(1,2,1),(1,2,2),(2,1,1),\ldots,(2,2,2)$, where vectors of empty intersections are skipped. We call the union of the cells added so far a cumulative set; there is one more cumulative set than there are cells, from $\emptyset$ to $V$. Two cumulative sets are consecutive if the second one arises from the first one by adding a single cell, so that they have the form $U_0$ and $U_0\cup A$ for a cell $A$. For disjoint vertex sets $X$ and $Y$, let $F^*(X,Y)$ denote the multiset of the copies of edges of $F^*$ between $X$ and $Y$. Every copy of an edge of $F^*$ between two different cells belongs to $F^*(U_0,A)$ for exactly one pair of consecutive cumulative sets $U_0$ and $U_0\cup A$, namely the one in which $A$ is the later of the two cells. Hence the multisets $F^*(U_0,A)$ over all pairs of consecutive cumulative sets together form exactly the copies of edges of $F^*$ between different cells, that is, $F^*_{\mathrm{small}}$. Guessing the edges of $F^*_{\mathrm{small}}$ thus reduces to guessing $F^*(U_0,A)$ for all pairs of consecutive cumulative sets.
 
Since the cells are not known, the dynamic program considers a family of sets that contains every cumulative set, whatever the small cuts are. Following Traub and Vygen, it uses boxed sets~\cite[Definition~16.19]{TraubVygen2024}, which are defined from the sets $U^i_j$ alone. A boxed set is given by $2\bar\kappa-1$ indices $0\le a_i<b_i\le n$ for $i<\bar\kappa$ and $1\le a_{\bar\kappa}\le n$, and it is the union of the disjoint sets
\[
\Bigl(\bigcap_{i'<i}\bigl(U^{i'}_{b_{i'}}\setminus U^{i'}_{a_{i'}}\bigr)\Bigr)\cap U^i_{a_i}\qquad\text{for }i\in\{1,\ldots,\bar\kappa\}.
\]
Every nonempty cumulative set is a boxed set: if $U_0$ and $U=U_0\cup A$ are consecutive cumulative sets and $U^i_{q_i}\setminus U^i_{p_i}$ is the slab of $A$ for each $i$, then, by the lexicographic order, $U$ is the boxed set with indices $a_i:=p_i$ and $b_i:=q_i$ for $i<\bar\kappa$ and $a_{\bar\kappa}:=q_{\bar\kappa}$. In words, the indices $p_i<q_i$ for $i<\bar\kappa$ delimit the slabs that are currently being filled, and the last index $q_{\bar\kappa}$ records how far they have been filled. For each $i<\bar\kappa$, there are at most $n(n+1)/2\le n^2$ choices of $a_i<b_i$, and there are $n$ choices of $a_{\bar\kappa}$, so there are at most $n^{2\bar\kappa-1}$ boxed sets. We let $\mathcal U$ consist of the empty set and all boxed sets. Thus $|\mathcal U|\le1+n^{2\bar\kappa-1}$, and $\mathcal U$ contains every cumulative set.
 
We next bound how $\mathcal W^*$ interacts with a cumulative set $U\ne\emptyset$, with indices as above. First, $F^*$ crosses $U$ only a few times. Indeed, $U$ is obtained from the $2\bar\kappa-1$ sets $U^i_{p_i}$ and $U^i_{q_i}$ for $i<\bar\kappa$ and $U^{\bar\kappa}_{q_{\bar\kappa}}$ by unions, intersections, and differences, and since $\delta(X\cup Y)$, $\delta(X\cap Y)$, and $\delta(X\setminus Y)$ lie in $\delta(X)\cup\delta(Y)$ for all $X,Y\subseteq V$, the cut $\delta(U)$ lies in the union of the cuts of these $2\bar\kappa-1$ sets. Since $p_i$ and $q_i$ are slab boundaries, each of these cuts is small or empty and hence contains at most $h$ copies of edges of $F^*$. Thus
\[
|F^*\cap\delta(U)|\le(2\bar\kappa-1)h.
\]
Second, $\mathcal W^*$ has only a few maximal subwalks inside $U$. A walk of $\mathcal W^*$ with $m$ edges in $\delta(U)$, counted with multiplicity, splits into $m+1$ maximal subwalks, which lie alternately in $U$ and outside $U$, so at most $\lceil(m+1)/2\rceil\le1+m/2$ of them lie in $U$. The start and end vertices of the maximal subwalks of $\mathcal W^*$ inside $U$, one pair for each subwalk, form the multiset $\mathcal T^*_U$ of terminal pairs induced by $\mathcal W^*$ on $U$. Summing over the $\kappa$ walks and using $\bar\kappa\le\kappa$, we obtain
\[
|\mathcal T^*_U|\le\kappa+\tfrac12|F^*\cap\delta(U)|\le\kappa+\bigl(\bar\kappa-\tfrac12\bigr)h\le\kappa(h+1),
\]
which also holds for $U=\emptyset$. These bounds show that little information has to be guessed.
 
We now describe the dynamic program. Let $\mathcal M_\beta$ be a $\beta$-approximation algorithm for the prize-collecting multi-path TSP that runs in polynomial time on instances whose number of terminal pairs is bounded by a constant. For $U\subseteq V$ and a multiset $\mathcal T'$ of terminal pairs in $U$, the local instance on $U$ with $\mathcal T'$ is the instance on $G[U]$ with the penalties of the input instance and the terminal pairs $\mathcal T'$, and the local objective value of a solution $\mathcal W'$ to it is $\ell(\mathcal W')+\pi_U(\mathcal W')$. For every set $U\in\mathcal U$, the dynamic program enumerates all multisets $\mathcal T'$ of at most $\kappa(h+1)$ terminal pairs in $U$, including the empty one, and computes a solution $Z[U,\mathcal T']$ to the local instance on $U$ with $\mathcal T'$. By the bound on $|\mathcal T^*_U|$, these include $\mathcal T^*_U$ for every cumulative set $U$. Thus, whereas Algorithm~\ref{alg:dp} enumerates local interfaces for each processed set, the dynamic program here enumerates multisets of terminal pairs. The sets in $\mathcal U$ are processed in increasing order of cardinality, $Z[\emptyset,\emptyset]$ is the empty family, and the algorithm returns $Z[V,\mathcal T]$, which is the second candidate; we call it the dynamic-programming candidate. For fixed $\kappa$ and $h$, the number of enumerated multisets $\mathcal T'$ is polynomial in $n$.
 
When a nonempty set $U\in\mathcal U$ is processed, all solutions $Z[U_0,\mathcal T_0]$ with $U_0\in\mathcal U$ and $U_0\subsetneq U$ have already been computed. For each enumerated $\mathcal T'$, the algorithm enumerates all sets $U_0\in\mathcal U$ with $U_0\subsetneq U$ and lets $A:=U\setminus U_0$. Thus $U_0$ plays the role of a child decomposition consisting of a single set, and $A$ that of the residual part.
 
For the chosen set $U_0$, the algorithm guesses how the walks of a solution on $U$ cross between $U_0$ and $A$. These guesses generalize the boundary multisets of Algorithm~\ref{alg:dp}: besides the edges between $U_0$ and $A$, they fix the order in which each walk traverses them. Formally, a border protocol for $U_0$ and a terminal pair $\{s,t\}$ in $U$, in an equivalent form of~\cite[Definition~16.16]{TraubVygen2024}, is a sequence $u_1u'_1,\ldots,u_mu'_m$ with $m\ge0$ of edges between $U_0$ and $A$ such that each of the $m+1$ vertex pairs
\[
\{s,u_1\},\ \{u'_1,u_2\},\ \ldots,\ \{u'_m,t\}
\]
lies in $U_0$ or in $A$; for $m=0$, the only such pair is $\{s,t\}$. We say that the protocol prescribes these $m+1$ pairs. Since every edge between $U_0$ and $A$ has one end vertex in $U_0$ and one in $A$, the order of the end vertices of each edge in a border protocol is determined by $s$: $u_1\in U_0$ if and only if $s\in U_0$, and $u_{j+1}\in U_0$ if and only if $u'_j\in U_0$. Thus any sequence of $m$ edges between $U_0$ and $A$, with the end vertices of its edges ordered in this way, is a border protocol for $\{s,t\}$ if and only if $m$ is even when $\{s,t\}$ lies in $U_0$ or in $A$, and odd otherwise. Given a border protocol, choosing for each prescribed pair a walk between its two vertices in $G[U_0]$ or in $G[A]$ and joining these walks by the edges $u_1u'_1,\ldots,u_mu'_m$, each traversed from $u_j$ to $u'_j$, gives a walk from $s$ to $t$ in $G[U]$; the protocol fixes the end vertices of these subwalks, but not the subwalks themselves.
 
Consider a multiset $\mathcal B$ consisting of one border protocol for each terminal pair in $\mathcal T'$. It specifies, for each terminal pair, which edges between $U_0$ and $A$ its walk traverses and in which order; the direction of each traversal is then determined, as explained above. We let $\delta_{\mathcal B}(U_0)$ be the multiunion of the edges of all protocols in $\mathcal B$; it keeps only the multiset of these edges and forgets how they are distributed among the terminal pairs and ordered. We also let $\mathcal T[U_0,\mathcal B]$ and $\mathcal T[A,\mathcal B]$ be the multisets of the pairs prescribed by all protocols in $\mathcal B$ that lie in $U_0$ and in $A$, respectively. For example, let $\mathcal T'=\{\{s_1,t_1\},\{s_2,t_2\},\{s_3,t_3\}\}$ with $s_1\in U_0$ and $t_1,s_2,t_2,s_3,t_3\in A$. Let $\mathcal B$ consist of the protocol $u_1u'_1$ for $\{s_1,t_1\}$, where $u_1\in U_0$ and $u'_1\in A$, the protocol $v_1v'_1,v_2v'_2$ for $\{s_2,t_2\}$, where $v_1,v'_2\in A$ and $v'_1,v_2\in U_0$, and the protocol without edges for $\{s_3,t_3\}$. Then
\begin{align*}
\mathcal T[U_0,\mathcal B]&=\{\{s_1,u_1\},\{v'_1,v_2\}\},\\
\mathcal T[A,\mathcal B]&=\{\{u'_1,t_1\},\{s_2,v_1\},\{v'_2,t_2\},\{s_3,t_3\}\},
\end{align*}
and walks joining these six pairs, together with the three edges of $\delta_{\mathcal B}(U_0)$, form walks from $s_i$ to $t_i$ for $i\in\{1,2,3\}$. In the terms of Algorithm~\ref{alg:dp}, $\mathcal T[U_0,\mathcal B]$ and $\mathcal T[A,\mathcal B]$ play the role of the local interfaces on the parts.
 
The algorithm enumerates all such multisets $\mathcal B$ whose protocols have at most $(2\bar\kappa-1)h$ edges in total and for which $\mathcal T[U_0,\mathcal B]$ has at most $\kappa(h+1)$ pairs. These are the bounds shown above for the cumulative sets. Following~\cite[Lemma~16.21]{TraubVygen2024}, the algorithm lists these multisets by choosing, for each terminal pair in $\mathcal T'$, a sequence of edges between $U_0$ and $A$ with the parity given above, which determines a border protocol, such that these sequences have at most $(2\bar\kappa-1)h$ edges in total, and keeping the choices for which $\mathcal T[U_0,\mathcal B]$ has at most $\kappa(h+1)$ pairs. Since $G$ has fewer than $n^2$ edges, there are at most $(n^2+1)^{(2\bar\kappa-1)h}$ choices of the sequence for each terminal pair, and $|\mathcal T'|\le\kappa(h+1)$, so, for fixed $\kappa$ and $h$, the number of enumerated multisets $\mathcal B$ is polynomial in $n$.

For each enumerated $\mathcal B$, let $\mathcal W_0:=Z[U_0,\mathcal T[U_0,\mathcal B]]$, which has been computed since $U_0\subsetneq U$ and $\mathcal T[U_0,\mathcal B]$ has at most $\kappa(h+1)$ pairs. The residual part $A$ is handled by running $\mathcal M_\beta$ on the local instance on $A$ with $\mathcal T[A,\mathcal B]$; let $\mathcal W_A$ be the returned family. Since each protocol prescribes at most one more pair in $A$ than in $U_0$, this local instance has $|\mathcal T[A,\mathcal B]|\le|\mathcal T[U_0,\mathcal B]|+|\mathcal T'|\le2\kappa(h+1)$ terminal pairs, so $\mathcal M_\beta$ runs on it in polynomial time for fixed $\kappa$ and $h$.

Joining the walks of $\mathcal W_0$ and $\mathcal W_A$ by the edges of $\delta_{\mathcal B}(U_0)$ along the protocols of $\mathcal B$, using each walk once, gives a solution to the local instance on $U$ with $\mathcal T'$. We verify this as in~\cite[Lemma~16.17]{TraubVygen2024}. Let $\{s,t\}$ be a terminal pair in $\mathcal T'$, and let $u_1u'_1,\ldots,u_mu'_m$ be its protocol in $\mathcal B$. Each of the prescribed pairs $\{s,u_1\},\{u'_1,u_2\},\ldots,\{u'_m,t\}$ lies in $U_0$ or in $A$. Since $\mathcal T[U_0,\mathcal B]$ and $\mathcal T[A,\mathcal B]$ are exactly the multisets of the prescribed pairs in $U_0$ and in $A$, over all protocols of $\mathcal B$, we can assign to each occurrence of a prescribed pair a distinct walk of $\mathcal W_0$ or of $\mathcal W_A$ between the two vertices of this pair. Traversing the walks assigned to $\{s,u_1\},\{u'_1,u_2\},\ldots,\{u'_m,t\}$ in this order, each in the appropriate direction, and joining consecutive ones by the edges $u_1u'_1,\ldots,u_mu'_m$ gives a walk from $s$ to $t$ in $G[U]$. Doing this for all terminal pairs in $\mathcal T'$ uses every walk of $\mathcal W_0$ and of $\mathcal W_A$ and every edge of $\delta_{\mathcal B}(U_0)$ exactly once. Hence the resulting family is a solution to the local instance on $U$ with $\mathcal T'$, its length is $\ell(\mathcal W_0)+\ell(\delta_{\mathcal B}(U_0))+\ell(\mathcal W_A)$, and it visits exactly the vertices visited by $\mathcal W_0$ or by $\mathcal W_A$.
 
Among all solutions obtained in this way, the dynamic program chooses one with minimum local objective value as $Z[U,\mathcal T']$. The choice $U_0=\emptyset$, with protocols without edges, always yields a solution, in which $\mathcal M_\beta$ is applied to the local instance on $U$ with $\mathcal T'$. Hence $Z[U,\mathcal T']$ is defined after $U$ is processed.
 
In contrast to~\cite{TraubVygen2024}, the local instances carry the penalties of the input instance, and a local instance without terminal pairs is feasible; in this case, $\mathcal W_A$ is the empty family and $\pi_A(\mathcal W_A)=\pi(A)$.
 
We now return to the analysis. Let $U_0$ and $U$ be consecutive cumulative sets, and let $A:=U\setminus U_0$. Each walk of $\mathcal W^*[U]$ induces a border protocol for $U_0$, namely the sequence of the edges between $U_0$ and $A$ that it traverses, in the order of traversal; let $\mathcal B^*$ be the multiset of these protocols. Together with the multisets $\mathcal T^*_U$ and $\mathcal T^*_{U_0}$ of terminal pairs induced by $\mathcal W^*$ on $U$ and on $U_0$, these protocols form the choices determined by $\mathcal W^*$ for $U$. By the bounds above, the dynamic program enumerates these choices, although it cannot identify them. Since every solution it computes is a cheapest one among all choices considered for it, its output is at most as expensive as the solution built along the cumulative sets with the choices determined by $\mathcal W^*$. In this solution, by the reduction above, the edges of $F^*_{\mathrm{small}}$ are paid for at their actual length, while the remaining edges and the penalties are paid for through $\mathcal M_\beta$. The next lemma makes this precise.
 
\begin{lemma}
\label{lem:mp-dp}
The dynamic-programming candidate $\mathcal W$ satisfies $\ell(\mathcal W)+\pi_V(\mathcal W)\le\beta\,\mathrm{OPT}-(\beta-1)\ell(F^*_{\mathrm{small}})$. The dynamic program makes polynomially many calls to $\mathcal M_\beta$, each on an instance with at most $2\kappa(h+1)$ terminal pairs, and otherwise runs in polynomial time, where the exponents depend only on $\kappa$ and $h$.
\end{lemma}
 
\begin{proof}
The running time follows from the bound on the number of boxed sets and the bounds on the numbers of enumerated multisets $\mathcal T'$ and $\mathcal B$ and on $|\mathcal T[A,\mathcal B]|$ given in the description of the dynamic program.
 
For the objective value, we follow the proofs of~\cite[Lemmas~16.22 and~16.23]{TraubVygen2024} and also account for the penalties. We show by induction along the sequence of cumulative sets that, for every cumulative set $U$, the dynamic program computes $Z[U,\mathcal T^*_U]$ and its local objective value is at most
\[
\ell(F^*_{\mathrm{small}}[U])+\beta\bigl(\ell(F^*[U])-\ell(F^*_{\mathrm{small}}[U])+\pi_U(\mathcal W^*)\bigr).
\]
The solution $Z[U,\mathcal T^*_U]$ is computed since $U\in\mathcal U$ and $|\mathcal T^*_U|\le\kappa(h+1)$. For $U=\emptyset$, the solution $Z[\emptyset,\emptyset]$ is the empty family, and both its local objective value and the bound are $0$. Otherwise, let $U_0$ be the previous cumulative set, let $A:=U\setminus U_0$, and let $\mathcal B^*$ be the multiset of border protocols induced by $\mathcal W^*[U]$ for $U_0$, as defined before the lemma. By the definition of $\mathcal B^*$, the multisets $\mathcal T[U_0,\mathcal B^*]$ and $\mathcal T[A,\mathcal B^*]$ consist of the pairs of end vertices of the walks obtained by restricting $\mathcal W^*[U]$ further to $U_0$ and to $A$, respectively. Since $U_0,A\subseteq U$, these restrictions are $\mathcal W^*[U_0]$ and $\mathcal W^*[A]$~\cite[Lemma~16.15]{TraubVygen2024}. Hence $\mathcal T[U_0,\mathcal B^*]=\mathcal T^*_{U_0}$, the multiset $\mathcal T[A,\mathcal B^*]$ consists of the terminal pairs induced by $\mathcal W^*$ on $A$, and $\delta_{\mathcal B^*}(U_0)=F^*(U_0,A)\subseteq F^*\cap\delta(U_0)$. Since $U_0\in\mathcal U$, $|\mathcal T^*_{U_0}|\le\kappa(h+1)$, and $|F^*\cap\delta(U_0)|\le(2\bar\kappa-1)h$, the dynamic program considers $U_0$ and $\mathcal B^*$ when computing $Z[U,\mathcal T^*_U]$. For this choice, it joins $\mathcal W_0=Z[U_0,\mathcal T^*_{U_0}]$ and the family $\mathcal W_A$ returned by $\mathcal M_\beta$ along the protocols of $\mathcal B^*$. The resulting solution has length $\ell(\mathcal W_0)+\ell(F^*(U_0,A))+\ell(\mathcal W_A)$ and visits every vertex visited by $\mathcal W_0$ or by $\mathcal W_A$, so its local objective value is at most the local objective value of $\mathcal W_0$ plus $\ell(F^*(U_0,A))+\ell(\mathcal W_A)+\pi_A(\mathcal W_A)$. Since $\mathcal W^*[A]$ is a solution to the local instance on $A$ with $\mathcal T[A,\mathcal B^*]$, $\ell(\mathcal W^*[A])=\ell(F^*[A])$, and $V(\mathcal W^*[A])=V(\mathcal W^*)\cap A$, the family $\mathcal W_A$ has local objective value at most $\beta\bigl(\ell(F^*[A])+\pi_A(\mathcal W^*)\bigr)$. Since $Z[U,\mathcal T^*_U]$ is a cheapest solution among the choices considered for it, the induction hypothesis for $U_0$ shows that its local objective value is at most
\[
\ell(F^*_{\mathrm{small}}[U_0])+\beta\bigl(\ell(F^*[U_0])-\ell(F^*_{\mathrm{small}}[U_0])+\pi_{U_0}(\mathcal W^*)\bigr)+\ell(F^*(U_0,A))+\beta\bigl(\ell(F^*[A])+\pi_A(\mathcal W^*)\bigr).
\]
Since an edge joins two different cells if and only if it lies in a small cut, the copies in $F^*(U_0,A)$ belong to $F^*_{\mathrm{small}}$, while those in $F^*[A]$ do not because $A$ is a single cell. As $F^*[U]=F^*[U_0]\uplus F^*(U_0,A)\uplus F^*[A]$, we obtain $F^*_{\mathrm{small}}[U]=F^*_{\mathrm{small}}[U_0]\uplus F^*(U_0,A)$, and hence
\begin{align*}
\ell(F^*_{\mathrm{small}}[U_0])+\ell(F^*(U_0,A))&=\ell(F^*_{\mathrm{small}}[U]),\\
\ell(F^*[U_0])-\ell(F^*_{\mathrm{small}}[U_0])+\ell(F^*[A])&=\ell(F^*[U])-\ell(F^*_{\mathrm{small}}[U]).
\end{align*}
Together with $\pi_U(\mathcal W^*)=\pi_{U_0}(\mathcal W^*)+\pi_A(\mathcal W^*)$, this shows that the bound above equals the claimed one.
 
For the last cumulative set, $U=V$ and $\mathcal T^*_V=\mathcal T$. Hence the output $Z[V,\mathcal T]$ of the dynamic program, which is the dynamic-programming candidate, has objective value at most
\[
\ell(F^*_{\mathrm{small}})+\beta\bigl(\ell(F^*)-\ell(F^*_{\mathrm{small}})+\pi_V(\mathcal W^*)\bigr)
=\beta\,\mathrm{OPT}-(\beta-1)\ell(F^*_{\mathrm{small}}).\qedhere
\]
\end{proof}
 
\subsection{Iterative improvement}
 
The following lemma is the analogue of Lemma~\ref{lem:one-step}.
 
\begin{lemma}
\label{lem:mp-improve}
Let $0<\eta\le1$, and let $\mathcal M_\beta$ be a $\beta$-approximation algorithm for the prize-collecting multi-path TSP. On an instance with $\kappa$ terminal pairs, let $h:=\lceil8\kappa/\eta\rceil$ if $\kappa\ge1$. Then the better of the short-distance candidate and, if $\bar\kappa\ge1$, the dynamic-programming candidate with $\mathcal M_\beta$ and this $h$ has objective value at most $\max\{\rho+\eta,\beta-\frac{\eta}{8}(\beta-1)\}\,\mathrm{OPT}$.
\end{lemma}
 
\begin{proof}
If $\ell(\mathcal T)\le(\eta/4)\mathrm{OPT}$, then, by Lemma~\ref{lem:mp-short}, the short-distance candidate has objective value at most $(\rho+\eta/2)\mathrm{OPT}+2\ell(\mathcal T)\le(\rho+\eta)\mathrm{OPT}$. Otherwise, $\ell(\mathcal T)>0$, so $\bar\kappa\ge1$, and Lemma~\ref{lem:mp-distance} together with $\kappa/(h+1)\le\eta/8$ and $\ell(F^*)\le\mathrm{OPT}$ gives $\ell(F^*_{\mathrm{small}})>(\eta/4-\eta/8)\mathrm{OPT}=(\eta/8)\mathrm{OPT}$. By Lemma~\ref{lem:mp-dp}, the dynamic-programming candidate has objective value at most $\bigl(\beta-\frac{\eta}{8}(\beta-1)\bigr)\mathrm{OPT}$.
\end{proof}
 
\begin{proof}[Proof of Corollary~\ref{cor:multipath}]
Fix $\varepsilon>0$ and set $\eta:=\min\{1,\varepsilon\}$. Let $\mathcal M_0$ be the algorithm that returns the short-distance candidate. Since $\ell(\mathcal T)\le\mathrm{OPT}$ and $\eta\le1$, it is a $\beta_0$-approximation algorithm for $\beta_0:=\rho+3$. For $j\ge1$, let $\mathcal M_j$ be the algorithm that returns the better of the two candidates in Lemma~\ref{lem:mp-improve} with $\mathcal M_\beta=\mathcal M_{j-1}$. By Lemma~\ref{lem:mp-improve} and induction on $j$, the algorithm $\mathcal M_j$ is a $\beta_j$-approximation algorithm, where
\[
\beta_j:=\max\left\{\rho+\eta,\ \beta_{j-1}-\frac{\eta}{8}(\beta_{j-1}-1)\right\}.
\]
Since $\beta-\frac{\eta}{8}(\beta-1)=1+(1-\frac{\eta}{8})(\beta-1)$ for every $\beta$, induction on $j$ gives
\[
\beta_j\le\max\left\{\rho+\eta,\ 1+\left(1-\frac{\eta}{8}\right)^j(\rho+2)\right\}.
\]
Let $N$ be the smallest integer with $(1-\eta/8)^N(\rho+2)\le\rho+\eta-1$; it depends only on $\rho$ and $\varepsilon$, and $\mathcal M_N$ has approximation ratio at most $\rho+\eta\le\rho+\varepsilon$. It remains to bound the running time. On instances without terminal pairs, every $\mathcal M_j$ returns the empty family. On instances with $1\le\kappa\le q$ terminal pairs, $\mathcal M_0$ runs the algorithm of Theorem~\ref{thm:main} with $k=2\kappa\le2q$, so $\mathcal M_0$ runs in polynomial time on instances whose number of terminal pairs is bounded by a constant. If $\mathcal M_{j-1}$ has this property, then so does $\mathcal M_j$: on instances with $1\le\kappa\le q$ terminal pairs, it additionally runs, if $\bar\kappa\ge1$, the dynamic program with $h=\lceil8\kappa/\eta\rceil\le\lceil8q/\eta\rceil$, which, by Lemma~\ref{lem:mp-dp}, makes polynomially many calls to $\mathcal M_{j-1}$, each on an instance with at most $2q(\lceil8q/\eta\rceil+1)$ terminal pairs, and otherwise runs in polynomial time. Since $N$ is a constant, $\mathcal M_N$ runs in polynomial time on instances whose number of terminal pairs is bounded by a constant, in particular for every fixed $\kappa$.
\end{proof}

\end{document}